\documentclass[11pt]{amsart}
\usepackage{amsmath,a4wide,graphicx,wasysym,amssymb,mathtools,bm}
\numberwithin{equation}{section}
\usepackage{tikz}
\usetikzlibrary{arrows,calc,shapes,decorations.pathreplacing,cd,patterns}
\usepackage[T1]{fontenc}
\usepackage[latin1]{inputenc}
\usepackage{hyperref}
\usepackage[hyperpageref]{backref}

\newtheorem{lemma}[equation]{Lemma}
\newtheorem{proposition}[equation]{Proposition}
\newtheorem{theorem}[equation]{Theorem}

\theoremstyle{definition}
\newtheorem{definition}[equation]{Definition}
\newtheorem{example}[equation]{Example}
\newtheorem{corollary}[equation]{Corollary}
\newtheorem{remark}[equation]{Remark}

\newcommand{\mb}[1]{{\mathbf #1}}
\newcommand{\mc}[1]{{\mathcal #1}}

\DeclareMathOperator*{\Coast}{\raisebox{-0.6ex}{\scalebox{2.5}{$\ast$}}}

\title{Connectivity of spaces of directed paths in geometric models for concurrent computation}
\author{Martin Raussen} 
\address{Department of
  Mathematical Sciences, Aalborg University, Skjernvej 4A,
  DK-9220 Aalborg {\O}st, Denmark} 
\email{raussen@math.aau.dk}
\thanks{Thanks to the anonymous referees for many hints leading to various improvements of the presentation.}
\begin{document}
\begin{abstract}
Higher Dimensional Automata (HDA) are higher dimensional relatives to transition systems in concurrency theory taking into account to which degree various actions commute. Mathematically, they take the form of labelled cubical complexes. It is important to know, and challenging from a geometric/topological perspective, whether the space of directed paths (executions in the model) between two vertices (states) is connected; more generally, to estimate higher connectivity of these path spaces.

This paper presents an approach for such an estimation for particularly simple HDA arising from PV programs and modelling the access of a number of processors to a number of resources with given limited capacity each. It defines the spare capacity of a concurrent program with prescribed periods of access of the processors to the resources using only the syntax of individual programs and the capacities of shared resources. It shows that the connectivity of spaces of directed paths can be estimated (from above) by spare capacities. Moreover, spare capacities can also be used to detect deadlocks and critical states in such a simple HDA.

The key theoretical ingredient is a transition from the calculation of local connectivity bounds (of the upper links of vertices of an HDA) to global ones by applying a version of the nerve lemma due to Anders Bj\"{o}rner. 
\end{abstract}

\maketitle
\section{PV programs and their geometric semantics}
\subsection{Euclidean cubical complexes and path spaces}\label{ss:Euclid}
In this paper, we consider an old and simple model for concurrent computation, the so-called PV programs going back to Dijkstra \cite{Dijkstra:68}: In the simplest situation, consider a set $P$ of $n$ processors $j\in P$ each executing a linear program. During execution, a processor may lock ($Pr$) and relinquish ($Vr$) one or several resources $r$ from a pool $R$ of resources; possibly several times consecutively. Any resource $r\in R$ has capacity to serve up to $\kappa (r)\in\mb{N}$ of these processors at any given time.

Every execution of a linear program for a given processor corresponds to a directed map $p_j: I=[0,1]\to I_j$ where $I_j$ is a finite closed interval; and directed means continuous and (not necessarily strictly) increasing. Disregarding conflicting locks in the first place, any concurrent execution corresponds to a (componentwise) directed path $p=(p_j):I\to \prod_{j=1}^nI_j\subseteq\mb{R}^n$. 

Lock and unlock commands occur consecutively along the intervals $I_j$, eg at integer values. They give rise to the \emph{forbidden region} $F\subset I_j\subseteq\mb{R}^n$ which is composed of products of subintervals representing states at which at least one of the resources is locked by more processors than its capacity allows. Relevant directed paths (aka d-paths in the literature) are not allowed to enter the forbidden region $F$. 
The ``allowed'' \emph{state space} $X\subset \mb{R}^n$ for such a PV program $\mc{P}$ is the complement of the forbidden region $F$ within $\prod_{j=1}^nI_j$. Our aim is to analyse properties of the space of directed paths in this state space (given a source and a target). For details cf Section \ref{ss:forbidden} or the more comprehensive description in Fajstrup etal. \cite[ch.\ 3]{FGHMR:16}. 

This state space is a very particular simple case of a pre-cubical set or rather its geometric realization (cf \cite[ch.\ 3.4]{FGHMR:16}, Ziemia\'{n}ski \cite[Sect.\ 1]{Ziemianski:17}, \cite[Sect.\ 2.4]{Ziemianski:20}) underlying general Higher Dimensional Automata introduced by Pratt \cite{Pratt} and analysed by van Glabbeek, cf eg \cite{Glabbeek:05}). It has a natural embedding into $\mb{R}^n$ (a Euclidean cubical complex); as such, it is non-self-linked and proper, cf Ziemia\'{n}ski \cite[Sect.\ 1]{Ziemianski:17}.

The space of all executions from an initial state $s$ to a final state $t$ in state space $X$ corresponds to the space $\vec{P}(X)_s^t$ of all \emph{directed} paths $p:I\to X$ with $p(0)=s$ and $p(1)=t$; equipped with the compact-open topology (uniform convergence). Various simplicial models of such spaces have (in far more general situations) been described in the literature (in particular Raussen \cite{Raussen:10, Raussen:12}, Fajstrup etal \cite {FGHMR:16} and Ziemia\'{n}ski \cite{Ziemianski:17,Ziemianski:20}).  The most notable result (Ziemia\'{n}ski \cite[Theorem 5.9]{Ziemianski:16}) reports that they can be arbitrarily complicated: For every finite simplicial complex $C$ there exists a PV program with a state space $X$ and vertices $s$ and $t$ such that the path space $\vec{P}(X)_s^t$ has a connected component that is homotopy equivalent to that complex. 

On the positive side, it was shown that, roughly speaking, path connectivity locally everywhere implies global path connectivity of all path spaces $\vec{P}(X)_s^t$ (cf Raussen \cite[Prop.\ 2.18]{Raussen:00}, Belton etal \cite[Theorem 2]{Beltonetal:19}). This is important: That the path space is connected means that any two directed paths are d-homotopic (homotopic through a one-parameter family of directed paths), and hence any two executions are equivalent and yield the same result (cf \cite {FGHMR:16} for details).

In this paper, we define the \emph{spare capacity of a PV program} $\mc{P}$ on a given resource pool. We show (Theorem \ref{prop:PVconn}) that it can be used as an estimate for the connectivity of spaces of concurrent executions corresponding to directed paths in the associated state space. If this spare capacity is at least $2$, the execution spaces are path-connected; if it is less than $2$, then non-equivalent executions may arise.

\subsection{A short overview} 
\subsubsection{The main lines}
The present paper yields a quite simple numerical criterion ensuring local connectivity of path spaces which, by an inductive argument using machinery from combinatorial topology (cf Bj\"{o}rner \cite[Theorem 6]{Bjoerner:03}), implies global connectivity. We call the number in question  the \emph{spare capacity} of the program, cf Definition \ref{def:cumcap}. The spare capacity allows to estimate not only path connectivity but also \emph{higher connectivity} of spaces of directed paths between vertices in the associated state space.

More precisely, take departure in a PV-program $\mc{P}$ on $n$ processors $j\in P$ with several shared resources $r$ within a pool $R$ of resources, each with a capacity $\kappa(r)\in\mb{N}$; meaning that up to $\kappa (r)$ processors may use resource $r$ at any given time. We determine, by a simple calculation, the spare capacity $\kappa(X)$ of the associated state space $X$, a numerical invariant that depends on the capacities of the resources and on the intersection patterns of intervals on which these resources are jointly called by the processors. We show (Theorem \ref{prop:PVconn}) that the spare capacity allows to determine the minimal (higher) connectivity of path spaces $\vec{P}(X)_s^t$ (with a fixed target $t$ that is reachable from a variable source $s$): There exists a vertex $s$ such that path space $\vec{P}(X)_s^t$ is exactly $(\kappa(X)-2)$-connected, and for all other vertices $s'$,  $\vec{P}(X)_{s'}^t$ is at least $(\kappa(X)-2)$-connected. 

In particular, if this invariant $\kappa(X)\ge 2$, it is guaranteed that all relevant (non-empty) path spaces are path-connected; ie any two directed paths from a source $s$ to a target $t$ are d-homotopic. That means that all executions of the concurrent program $X$ (same source, same target, same individual execution along each thread, in particular same number of loops traversed) yield the \emph{same result} regardless the order of accesses to shared resources. Such a program can thus \emph{not} be used to solve a decision problem (cf eg Herlihy etal \cite{HKR:14}). On the other hand, if $\kappa(X)=1$, then there exists a vertex and directed paths starting from that vertex which are not dihomotopic. Hence, corresponding executions may lead to \emph{different results}; such a program might then be useful for solving a decision problem. 

Throughout most of this paper, we stick to concurrent executions of \emph{linear} programs. This might seem very restrictive and unrealistic; processors usually execute programs on \emph{directed graphs}, allowing branchings and loops. The space of all such executions between a source and a target decomposes (by unfolding) into a disjoint union of spaces of executions along the linear directed paths within such a directed graph. As a consequence, the space of all executions in a concurrent program splits into spaces of executions along $n$-tuples of such linear directed paths; for these linear subspaces, the methods developed below apply. For details, consult Section \ref{ss:braloo}.

\subsubsection{Content of the paper in more detail}
In Section 2, we describe the point of departure: Given a number of processors, each following a linear program and participating in a concurrent program with shared resources. Each resource has a given capacity and may be locked and relinquished sequentially by every processor executing its program. The geometric semantics corresponds to a state space in the form of a Euclidean cubical complex, a subcomplex of a cubical subdivision of $\mb{R}^n$. Executions correspond to \emph{directed} paths from a source vertex to a target vertex. The aim is to study the \emph{space} of all such executions as a topological space; in particular to determine its connectivity. This section recapitulates in essence the point of departure in Fajstrup etal \cite[ch.\ 3]{FGHMR:16}.

Section 3 focusses on the \emph{local} behaviour of the state space. As already explained by Ziemia\'{n}ski \cite{Ziemianski:16} and Belton etal \cite{Beltonetal:19,Beltonetal:21}, the key information is the topology (in particular, the connectivity) of the \emph{future links} (or past links) of vertices in the state space. It turns out (Proposition \ref{prop:linkjoin}) that these future links are \emph{joins} (aka convex combinations) of \emph{skeleta of simplices}. This observation lets us determine the connectivity of the future link of a vertex $v$ in terms of a spare capacity $\kappa (X;v)$ \emph{defined at that vertex} $v$; this spare capacity can be calculated from the syntax of the individual programs by a combinatorial formula; cf Definition \ref{def:cc}. Roughly speaking, the spare capacity at a vertex $v$ expresses the number of processors that can proceed from $v$ concurrently.

In Section 4, we define the \emph{spare capacity} $\kappa (X)$ of a concurrent program $\mc{P}$ (or its state space $X$) as the \emph{minimum} of the spare capacities of all its reachable vertices; cf Definition \ref{def:cumcap}. Bj\"{o}rner's version of the nerve theorem \cite[Theorem 6]{Bjoerner:03} is then applied to conclude that the connectivity of the state space is bounded below by the connectivities of the future links of all in-between vertices, and hence by the spare capacity of the program, cf Theorem \ref{prop:PVconn}: If $\kappa (X)\ge 2$, then all directed paths from a start vertex to a target vertex in the state space are dihomotopic to each other; corresponding executions will always lead to the same result. In particular cases, the spare capacity bound is tight, cf Proposition \ref{prop:mini}. 

The overall result holds also for concurrent programs  consisting of individual non-linear individual programs modelled on a general digraph and its unfoldings. Section 4  contains, moreover, reflections on what happens to spare capacities (and thus connectivities) if processors are allowed to crash, cf Proposition \ref{prop:vertexcrash}.

The final section is devoted to a sketch of algorithmic aspects concerning the calculation of this spare capacity. Particular care is devoted to deadlock detection -- corresponding to spare capacity $0$ at a vertex -- extending the results of Fajstrup etal \cite{FGR:98}. Moreover to vertices with spare capacity $1$ indicating potential ``splits'' of the space of executions into several path components. Throughout, simple examples and illustrations motivate the strategy.

\section{Forbidden region and state space}
Throughout the paper $[m:n]$ denotes the set of integers between integers $m$ and $n$, whereas $[a,b]$, resp.\ $]a,b[$ denote the closed, resp.\ open intervals between real numbers $a$ and $b$.
\subsection{Resource consumption}\label{ss:consump}
Let $R$ denote the set of resources and $P=[1:n]$ the set of processors.
\subsubsection{One processor} A $PV$ command line for a single processor $j\in P$ (we take only account of the lock and release commands, not of the calculations taking part inbetween) indicates at which places processor $j$ requires a lock to $r$ (by issuing $Pr$), resp.\ relinquishes it (by issuing $Vr$); cf \cite{Dijkstra:68}. It can be encoded by a total of $2|R|$ functions\\ $Pr_j, Vr_j:[1:k_j(r)]\to [1:l(j)],\; r\in R$, such that 
\begin{itemize}
\item $Pr_j(i)<Vr_j(i),\; i\le k_j(r)$, and $Vr_j(i)<Pr_j(i+1),\; i<k_j(r)$, and \item $\bigcup_{r\in R, i\in [1:k_j(r)]}\{ Pr_j(i), Vr_j(i)\}=[1:l(j)]$.
\end{itemize}
Above, the number of lock and relase commands to a specific resource $r$ is denoted $k_j(r)$. It is allowed that $k_j(r)=0$, ie that some resources are not called upon by processor $j$. The total number of lock and release commands issued by $j$, to all resources $r\in R$, is denoted $l(j)$.  

\begin{definition}\label{consumption}
\begin{enumerate}
\item For $r\in R$ and $j\in P$, let $cr_j: I_J:=[0,l(j)+1]\to\{ 0,1\}$ denote the characteristic function of the subset $\bigcup_{i\in [1:k_j(r)]}]Pr_j(i),Vr_j(i)[\cup \{ l(j)+1\}$ indicating whether $j$ has a lock to $r$ or not (or has arrived at the final state). The characteristic functions for all $r\in R$ assemble to a binary valued \emph{consumption vector function} $c_j: [0,l(j)+1]\to\{ 0,1\}^R$. 

\item Furthermore, we let $dr_j: [0:l(j)+1]\to\{ -1,0,1\}$ denote the difference of the characteristic functions of the two integer sets\\ $\{ Pr_j(i)|\; 1\le i\le k_j(r)\}$ and $\{ Vr_j(i)|\; 1\le i\le k_j(r)\}$. In more detail, $d r_j(k)=\pm 1$ if $k=Pr_j(i)$, resp.\ $k=Vr_j(i)$ for some $i\in [1:k_j(r)]$, and $0$ else. These functions assemble to a vector function $d_j: [0:l(j)+1]\to\{ -1,0,1\}^R$ encompassing changes to locks to resources due to processor $j$.
\end{enumerate}  
\end{definition}

\begin{remark}\label{rem:cd}
\begin{enumerate}
\item Remark, that the intervals in the definition of the consumption function are open! 
\item The two functions in Definition \ref{consumption} are linked by \[c_j(i+t)=c_j(i)+dr_j(i),\; i\in [1:l(j)-1], t\in ]0,1[.\]
\end{enumerate}
\end{remark}

\subsubsection{Several processors}\label{sss:several} Information regarding consumption of resources by \emph{all} processors $j\in P$ is encoded by functions on $\prod_{j\in P}[0,l(j)+1)]\subset\mb{R}^P$: 

\begin{definition}\label{def:condiff}
The total \emph{consumption vector function}  
$\mb{c}:\prod_{j\in P}[0,l(j)+1)]\to (\mb{N}_{\ge 0})^R$ and the total \emph{difference vector function} $\mb{d}:\prod_{j\in P}[0:l(j)+1)]\to [-n:n]^R$  are defined by
\begin{align*}
 \mb{c}(x_1,\dots ,x_n)  &= \sum_{j\in P}c_j(x_j)\\
\mb{d}(x_1,\dots ,x_n)  &= \sum_{j\in P}d_j(x_j).
 \end{align*}
 \end{definition}

The consumption function measures how many locks to resources $r\in R$ have been acquired at $(x_1,\dots ,x_n)$. They have component functions $cr:\prod_{j\in P}[0,l(j)+1)]\to (\mb{N}_{\ge 0})$, resp.\ $dr:\prod_{j\in P}[0:l(j)+1)]\to [-n:n]$ for every $r\in R$. 

How does the consumption function change when proceeding from an integer vertex $v=(i_1,\dots ,i_n)\in \prod_{j\in P}[0:l(j)]\cap\mb{Z}^n$? From Remark \ref{rem:cd}(2), we conclude that
 \[cr (i_1+t_1,\dots ,i_n+t_n)=cr(i_1,\dots ,i_n)+\sum_{t_j>0}dr_j(i_j),\; 0\le t_j<1.\] This last sum encodes the difference between the number of requests to and the number of releases to $r$ at the vertex $v=(i_1,\dots ,i_n)$ on a given set of processors (those $j$ with $t_j>0)$.

\begin{example}\label{ex:cons}
\begin{enumerate}
\item The iconic Swiss flag example concerns a concurrent program on two processors $P:=\{1, 2\}$ with programs sharing two resources $a$ and $b$, both with capacity $1$ (mutual exclusion), called upon as $PaPbVbVa$ by $1$ resp.\ $PbPaVaVb$ by $2$. Figure \ref{fig:Swiss} illustrates the associated consumption function $\mb{c}$ in the interior of squares, the difference functions $d_1,d_2$ (Definition \ref{def:condiff}; $\mb{d}(v_1,v_2)=d_1(v_1)+d_2(v_2)$ at vertices $v=(v_1,v_2)$), the forbidden region $F$ (cf Section \ref{ss:forbidden}; in pink) and the state space $X$ as its complement. 

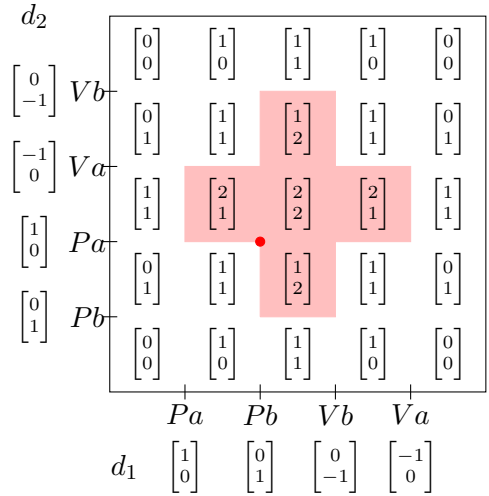
\begin{figure}[h]
\begin{tikzpicture}
\draw (0,0) rectangle (5,5);
\draw (1,-0.1) -- (1,0.1);
\draw (2,-0.1) -- (2,0.1);
\draw (3,-0.1) -- (3,0.1);
\draw (4,-0.1) -- (4,0.1);
\draw (-0.1,1) -- (0.1,1);
\draw (-0.1,2) -- (0.1,2);
\draw (-0.1,3) -- (0.1,3);
\draw (-0.1,4) -- (0.1,4);
\node at (1,-0.3) {$Pa$};
\node at (2,-0.3) {$Pb$};
\node at (3,-0.3) {$Vb$};
\node at (4,-0.3) {$Va$};
\node at (-0.3,1) {$Pb$};
\node at (-0.3,2) {$Pa$};
\node at (-0.3,3) {$Va$};
\node at (-0.3,4) {$Vb$};
\draw[fill,color=pink] (2,1) -- (3,1) -- (3,2) -- (4,2) -- (4,3) -- (3,3) -- (3,4) -- (2,4) -- (2,3) -- (1,3) -- (1,2) -- (2,2) -- (2,1);
\node at (0.5,0.5) {$\tiny{\begin{bmatrix} 0 \\ 0\end{bmatrix}}$};
\node at (1.5,0.5) {$\tiny{\begin{bmatrix} 1 \\ 0\end{bmatrix}}$};
\node at (2.5,0.5) {$\tiny{\begin{bmatrix} 1 \\ 1\end{bmatrix}}$};
\node at (3.5,0.5) {$\tiny{\begin{bmatrix} 1 \\ 0\end{bmatrix}}$};
\node at (4.5,0.5) {$\tiny{\begin{bmatrix} 0 \\ 0\end{bmatrix}}$};
\node at (0.5,1.5) {$\tiny{\begin{bmatrix} 0 \\ 1\end{bmatrix}}$};
\node at (1.5,1.5) {$\tiny{\begin{bmatrix} 1 \\ 1\end{bmatrix}}$};
\node at (2.5,1.5) {$\tiny{\begin{bmatrix} 1 \\ 2\end{bmatrix}}$};
\node at (3.5,1.5) {$\tiny{\begin{bmatrix} 1 \\ 1\end{bmatrix}}$};
\node at (4.5,1.5) {$\tiny{\begin{bmatrix} 0 \\ 1\end{bmatrix}}$};
\node at (0.5,2.5) {$\tiny{\begin{bmatrix} 1 \\ 1\end{bmatrix}}$};
\node at (1.5,2.5) {$\tiny{\begin{bmatrix} 2 \\ 1\end{bmatrix}}$};
\node at (2.5,2.5) {$\tiny{\begin{bmatrix} 2 \\ 2\end{bmatrix}}$};
\node at (3.5,2.5) {$\tiny{\begin{bmatrix} 2 \\ 1\end{bmatrix}}$};
\node at (4.5,2.5) {$\tiny{\begin{bmatrix} 1 \\ 1\end{bmatrix}}$};
\node at (0.5,3.5) {$\tiny{\begin{bmatrix} 0 \\ 1\end{bmatrix}}$};
\node at (1.5,3.5) {$\tiny{\begin{bmatrix} 1 \\ 1\end{bmatrix}}$};
\node at (2.5,3.5) {$\tiny{\begin{bmatrix} 1 \\ 2\end{bmatrix}}$};
\node at (3.5,3.5) {$\tiny{\begin{bmatrix} 1 \\ 1\end{bmatrix}}$};
\node at (4.5,3.5) {$\tiny{\begin{bmatrix} 0 \\ 1\end{bmatrix}}$};
\node at (0.5,4.5) {$\tiny{\begin{bmatrix} 0 \\ 0\end{bmatrix}}$};
\node at (1.5,4.5) {$\tiny{\begin{bmatrix} 1 \\ 0\end{bmatrix}}$};
\node at (2.5,4.5) {$\tiny{\begin{bmatrix} 1 \\ 1\end{bmatrix}}$};
\node at (3.5,4.5) {$\tiny{\begin{bmatrix} 1 \\ 0\end{bmatrix}}$};
\node at (4.5,4.5) {$\tiny{\begin{bmatrix} 0 \\ 0\end{bmatrix}}$};
\node at (0.2,-1) {$\tiny{d_1}$};
\node at (1,-1) {$\tiny{\begin{bmatrix}1\\ 0\end{bmatrix}}$};
\node at (2,-1) {$\tiny{\begin{bmatrix}0\\ 1\end{bmatrix}}$};
\node at (3,-1) {$\tiny{\begin{bmatrix}0\\ -1\end{bmatrix}}$};
\node at (4,-1) {$\tiny{\begin{bmatrix}-1\\ 0\end{bmatrix}}$};
\node at (-1,5) {$\tiny{d_2}$};
\node at (-1,4) {$\tiny{\begin{bmatrix}0\\ -1\end{bmatrix}}$};
\node at (-1,3) {$\tiny{\begin{bmatrix}-1\\ 0\end{bmatrix}}$};
\node at (-1,2) {$\tiny{\begin{bmatrix}1\\ 0\end{bmatrix}}$};
\node at (-1,1) {$\tiny{\begin{bmatrix}0\\ 1\end{bmatrix}}$};
\draw[fill,color=red] (2,2) circle (0.06cm);
\end{tikzpicture}
\caption{Swiss flag with associated consumption function: forbidden region in pink. The red vertex represents a deadlock: no non-trivial directed path in the state space starts from that vertex.}
\end{figure}\label{fig:Swiss}
 \item 
Consider the case of three processors $1,2,3\in P$ each of which executes a $PV$ program $PrVrPrVr$ on one resource $r$. The consumption function $cr$ takes values $cr(x_1,x_2,x_3)=0,1,2,3$ depending on how many of the coordinates are properly sandwiched between an odd and an even integer. In particular $cr(v)=0$ at \emph{every vertex} $v$. At a vertex $v=(i_1,i_2,i_3)$, we have $d_j(i_j)=1$ if $i_j$ is odd, $d_j(i_j)=-1$ if $i_j$ is even, and $d_j(i_j)=0$ if $i_j=0$ or $i_j=l(j)+1$. The associated forbidden regions (Section \ref{ss:forbidden}) for $\kappa (r)= 1$, resp.\ $\kappa (r)=2$ is illustrated in Figure \ref{fig:PV}.
\item Consider two resources $r, s$ of capacity three each, and four processors $i\in [1:4]$. 
\begin{description}
\item[A] Processors $i, i\le 3,$ start with $PrPs$ and processor $4$ starts with $PsPr$. At the vertex $v=(2,2,2,2)\in X_0^P$ with the final lock requests, $cr(v)=3, cs(v)=1, dr(v)=1, ds(v)=3$. 
\item[B] $1,2$ start with $PrPs$ and $3,4$ starts with $PsPr$. At the vertex $v=(2,2,2,2)$, $cr(v)=cs(v)=dr(v)=ds(v)=2$.
\end{description}
\item Now add a fifth processor $0$ starting with $PrPs$. 
\begin{description}
\item[A] The vertex $w=(2,2,2,2,2)$ is forbidden, cf Section \ref{ss:forbidden}: $cr(w)=4>3=\kappa (r)$.
\item[B] In this case, $cr(w)=3, dr(w)=2, cs(w)=2, ds(w)=3$. 
\end{description}
\end{enumerate}
\end{example}

\subsection{Forbidden region. State space}\label{ss:forbidden}
\subsubsection{A single shared resource}\label{sss:one}
We start by considering the case of a concurrent program in which processors in a set $P:=\{ 1,\dots ,n\}$ compete for a \emph{single} resource $r$ with capacity $\kappa (r)$ called upon (often several times) by programs each of the form $(PrVr)^{k_j},\; 1\le j\le n, k_j\ge 0$. The functions $Pr_j, Vr_j: [1:k_j]\to [1:l(j)=2k_j]$ are then given by $Pr_j(i)=2i-1$ and $Vr(i)=2i,\; 1\le i\le k_j$. The case $k_j=0$ takes take care of processors that do not call on $r$ at all. Let $N(r)\subseteq [1:n]$ denote the subset of processors with the property: $j\in N(r)\Leftrightarrow k_j>0$, ie processor $j$ calls upon $r$ at least once.

The corresponding \emph{forbidden} region $F(r)$, expressing that consumption exceeds the capacity of $r$, is defined as 
\[F(r):=\{ \mb{x}=(x_1,\dots ,x_n)\in \prod_{j\in P}[0,l(j)+1)]|\; cr(\mb{x})>\kappa:=\kappa (r)\}.\]
 It can be described as
a union of subsets, each a product of intervals, and enumerated as follows: Consider any injection $i: [1:\kappa +1]\hookrightarrow N(r)\subseteq [1:n]$ (ie a choice of $\kappa +1$ active processors) and the dual projection $i^*:\mb{R}^n\to\mb{R}^{\kappa +1}$. For every such injection consider all $(\kappa +1)$-tuples $\mb{l}:=(j_1,\dots , j_{\kappa +1})$ such that $0<j_m\le j(i(r)(m)), \; 1\le m\le \kappa +1$, enumerating all combinations of  ``lock intervals'' for the choice of processors given by $i$. For each combined choice $(i,\mb{l})$, let $F(i,\mb{l}):=(i^*)^{-1}(\prod_{m=1}^{\kappa +1}]Pr_{i(m)}(j_m),Vr_{i(m)}(j_m)[)$; a product that has $(n-\kappa -1)$ factors consisting of an entire interval $[0,l_j+1]$ corresponding to every $ j\not\in i(r)([1:\kappa +1])$.\\ The entire forbidden region is then $F(r)=\bigcup_{(i,\mb{l})}F(i,\mb{l})$. Figure \ref{fig:PV} shows the forbidden regions associated to Example \ref{ex:cons}(2).

The \emph{state space} $X(r)$ is the complement of the forbidden region: 
\[X(r):=\prod_{j=1}^n[0, l(k_j)+1]\setminus F(r).\] 

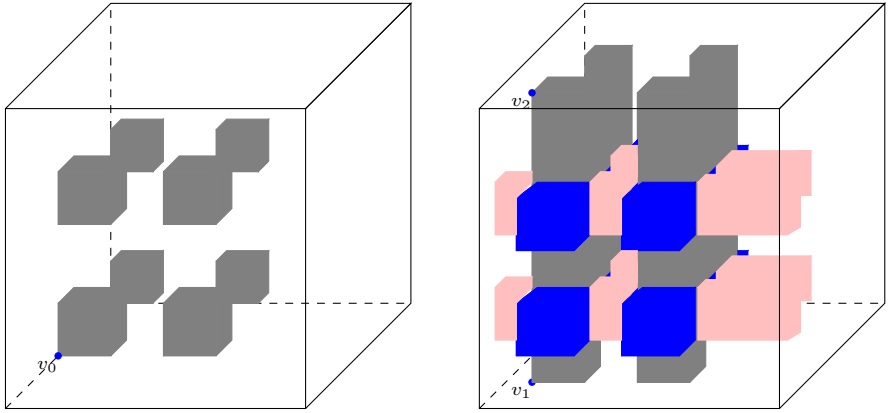
\begin{figure}[h]
\begin{tikzpicture}[scale=0.7]
\draw (0,0) rectangle (5.7,5.7);
\draw (0,5.7) -- (2,7.7) -- (7.7,7.7) -- (5.7,5.7);
\draw (5.7,0) -- (7.7,2) -- (7.7,7.7) -- (5.7,5.7);
\draw[dashed] (0,0) -- (1,1);
\draw[dashed] (2,2) -- (7.7,2);
\draw[dashed] (2,2) -- (2,7.7);
\draw[fill,color=gray] (1,1) rectangle (2,2);
\draw[fill,color=gray] (1,2) -- (1.3,2.3) -- (2.3,2.3) -- (2,2) -- (1,2);
\draw[fill,color=gray] (2,1) -- (2,2) -- (2.3,2.3) -- (2.3,1.3) -- (2,1);
\draw[fill,color=gray] (3,1) rectangle (4,2);
\draw[fill,color=gray] (3,2) -- (3.3,2.3) -- (4.3,2.3) -- (4,2) -- (3,2);
\draw[fill,color=gray] (4,1) -- (4,2) -- (4.3,2.3) -- (4.3,1.3) -- (4,1);
\draw[fill,color=gray] (2,2) rectangle (2.8,2.8);
\draw[fill,color=gray] (2,2.8) -- (2.2,3) -- (3,3) -- (2.8,2) -- (2,2.8);
\draw[fill,color=gray] (2.8,2) -- (2.8,2.8) -- (3,3) -- (3,2.2) -- (2.8,2);
\draw[fill,color=gray] (4,2) rectangle (4.8,2.8);
\draw[fill,color=gray] (4,2.8) -- (4.2,3) -- (5,3) -- (4.8,2) -- (4,2.8);
\draw[fill,color=gray] (4.8,2) -- (4.8,2.8) -- (5,3) -- (5,2.2) -- (4.8,2);
\draw[fill,color=gray] (1,3.5) rectangle (2,4.5);
\draw[fill,color=gray] (1,4.5) -- (1.3,4.8) -- (2.3,4.8) -- (2,4.5) -- (1,4.5);
\draw[fill,color=gray](2,3.5) -- (2,4.5) -- (2.3,4.8) -- (2.3,3.8) -- (2,3.5);
\draw[fill,color=gray] (3,3.5) rectangle (4,4.5);
\draw[fill,color=gray] (3,4.5) -- (3.3,4.8) -- (4.3,4.8) -- (4,4.5) -- (3,4.5);
\draw[fill,color=gray] (4,3.5) -- (4,4.5) -- (4.3,4.8) -- (4.3,3.8) -- (4,3.5);
\draw[fill,color=gray] (2,4.5) rectangle (2.8,5.3);
\draw[fill,color=gray] (2,5.3) -- (2.2,5.5) -- (3,5.5) -- (2.8,4.5) -- (2,5.3);
\draw[fill,color=gray] (2.8,4.5) -- (2.8,5.3) -- (3,5.5) -- (3,4.7) -- (2.8,4.5);
\draw[fill,color=gray] (4,4.5) rectangle (4.8,5.3);
\draw[fill,color=gray] (4,5.3) -- (4.2,5.5) -- (5,5.5) -- (4.8,4.5) -- (4,5.3);
\draw[fill,color=gray] (4.8,4.5) -- (4.8,5.3) -- (5,5.5) -- (5,4.7) -- (4.8,4.5);
\draw[fill,color=blue] (1,1) circle (0.06cm);
\draw[fill,color=blue] (10,0.5) circle (0.06cm);
\node at (0.8,0.8) {\tiny{$v_0$}};
\node at (9.8,0.3) {\tiny{$v_1$}};
\draw[dashed] (9,0) -- (10,1);
\draw[dashed] (11,2) -- (16.7,2);
\draw[dashed] (11,2) -- (11,7.7);
\draw[fill,color=gray] (10,0.5) rectangle (11,6);
\draw[fill,color=gray] (10,6) -- (10.3,6.3) -- (11.3,6.3) -- (11,6) -- (10,6);
\draw[fill,color=gray] (11,0.5) -- (11.3,0.8) -- (11.3,6.3) -- (11,6) -- (11,0.5); 
\draw[fill,color=gray] (12,0.5) rectangle (13,6);
\draw[fill,color=gray] (12,6) -- (12.3,6.3) -- (13.3,6.3) -- (13,6) -- (12,6);
\draw[fill,color=gray] (13,0.5) -- (13.3,0.8) -- (13.3,6.3) -- (13,6) -- (13,0.5); 
\draw[fill,color=gray] (11,1.5) rectangle (11.7,6.7);
\draw[fill,color=gray] (11,6.7) -- (11.2,6.9) -- (11.9,6.9) -- (11.7,6.7) -- (11,6.7);
\draw[fill,color=gray] (11.7,1.5) -- (11.9,1.7) -- (11.9,6.9) -- (11.7,6.7) -- (11.7,1.5); 
\draw[fill,color=gray] (13,1.5) rectangle (13.7,6.7);
\draw[fill,color=gray] (13,6.7) -- (13.2,6.9) -- (13.9,6.9) -- (13.7,6.7) -- (13,6.7);
\draw[fill,color=gray] (13.7,1.5) -- (13.9,1.7) -- (13.9,6.9) -- (13.7,6.7) -- (13.7,1.5); 
\draw[fill, color=blue] (9.7,1) rectangle (10.8,2);
\draw[fill, color=blue] (9.7,2) -- (10.1,2.3) -- (11.1,2.3) -- (10.8,2) -- (9.7,2);
\draw[fill, color=blue] (10.8,1) -- (11.1,1.3) -- (11.1,2.3) -- (10.8,2) -- (10.8,1);
\draw[fill,color=blue] (11.3,1.5) -- (11.5,1.7) -- (11.5,2.6) -- (11.3,2.4) -- (11.3,1.5);
\draw[fill, color=blue] (11.7,1) rectangle (12.8,2);
\draw[fill, color=blue] (11.7,2) -- (12.1,2.3) -- (13.1,2.3) -- (12.8,2) -- (11.7,2);
\draw[fill, color=blue] (12.8,1) -- (13.1,1.3) -- (13.1,2.3) -- (12.8,2) -- (12.8,1);
\draw[fill,color=blue] (13.3,1.5) -- (13.5,1.7) -- (13.5,2.6) -- (13.3,2.4) -- (13.3,1.5);
\draw[fill, color=blue] (9.7,3) rectangle (10.8,4);
\draw[fill, color=blue] (9.7,4) -- (10.1,4.3) -- (11.1,4.3) -- (10.8,4) -- (9.7,4);
\draw[fill, color=blue] (10.8,3) -- (11.1,3.3) -- (11.1,4.3) -- (10.8,4) -- (10.8,3);
\draw[fill,color=blue] (11.3,3.5) -- (11.5,3.7) -- (11.5,4.6) -- (11.3,4.4) -- (11.3,3.5);
\draw[fill, color=blue] (11.7,3) rectangle (12.8,4);
\draw[fill, color=blue] (11.7,4) -- (12.1,4.3) -- (13.1,4.3) -- (12.8,4) -- (11.7,4);
\draw[fill, color=blue] (12.8,3) -- (13.1,3.3) -- (13.1,4.3) -- (12.8,4) -- (12.8,3);
\draw[fill,color=blue] (13.3,3.5) -- (13.5,3.7) -- (13.5,4.6) -- (13.3,4.4) -- (13.3,3.5);
\draw[fill,color=blue] (13.9,2) -- (14.1,2.2) -- (14.1,3) -- (13.9,2.8) -- (13.9,2.2);
\draw[fill,color=blue] (13.9,2.8) -- (14.1,3) -- (13.9,3);
\draw[fill,color=blue] (13.9,4) -- (14.1,4.2) -- (14.1,5) -- (13.9,4.8) -- (13.9,4.2);
\draw[fill,color=blue] (13.9,4.8) -- (14.1,5) -- (13.9,5);
\draw[fill,color=blue] (11.9,2.1) -- (12,2.2) -- (12,3.1) -- (11.9,3) -- (11.9,2.1);
\draw[fill,color=blue] (11.9,4.1) -- (12,4.2) -- (12,5.1) -- (11.9,5) -- (11.9,4.1);
\draw[fill,color=blue] (11.9,5) -- (12,5.1) -- (12,5);
\draw[fill,color=pink] (13.15,1.3) rectangle (14.85,2.3);
\draw[fill,color=pink] (14.85,1.3) -- (15.1,1.45) -- (15.1,2.45) -- (14.85,2.3) -- (14.85,1.3);
\draw[fill,color=pink] (14.85,2.3) -- (15.1,2.45) -- (13.3,2.45) -- (13.15,2.3) -- (14.85,2.3);
\draw[fill,color=pink] (11.1,1.3) rectangle (11.7,2.3);
\draw[fill,color=pink] (11.1,2.3) -- (11.3,2.5) -- (12,2.5) -- (11.8,2.3) -- (11.1,2.3);
\draw[fill,color=pink] (11.7,2) -- (11.7,2.3) -- (11.8,2.3) -- (12,2.5) -- (12,2.3) -- (11.7,2);
\draw[fill,color=pink] (9.3,1.3) rectangle (9.7,2.3);
\draw[fill,color=pink] (9.3,2.3) -- (9.5,2.5) -- (10,2.5) -- (9.8,2.3) -- (9.3,2.3);
\draw[fill,color=pink] (9.8,2.3) -- (10,2.5) -- (10,2.3) -- (9.7,2) -- (9.7,2.3);
\draw[fill,color=pink] (13.15,3.3) rectangle (14.85,4.3);
\draw[fill,color=pink] (14.85,3.3) -- (15.1,3.45) -- (15.1,4.45) -- (14.85,4.3) -- (14.85,3.3);
\draw[fill,color=pink] (14.85,4.3) -- (15.1,4.45) -- (13.3,4.45) -- (13.15,4.3) -- (14.85,4.3);
\draw[fill,color=pink] (11.1,3.3) rectangle (11.7,4.3);
\draw[fill,color=pink] (11.1,4.3) -- (11.3,4.5) -- (12,4.5) -- (11.8,4.3) -- (11.1,4.3);
\draw[fill,color=pink] (11.7,4) -- (11.7,4.3) -- (11.8,4.3) -- (12,4.5) -- (12,4.3) -- (11.7,4);
\draw[fill,color=pink] (9.3,3.3) rectangle (9.7,4.3);
\draw[fill,color=pink] (9.3,4.3) -- (9.5,4.5) -- (10,4.5) -- (9.8,4.3) -- (9.3,4.3);
\draw[fill,color=pink] (9.8,4.3) -- (10,4.5) -- (10,4.3) -- (9.7,4) -- (9.7,4.3);
\draw[fill,color=pink] (13.5,2.4) rectangle (15,2.6);
\draw[fill,color=pink] (13.5,2.6) -- (13.8,2.9) -- (15.3,2.9) -- (15,2.6) -- (13.5,2.6);
\draw[fill,color=pink] (15.3,2.9) -- (15.3,2.05) -- (15,2.05) -- (15,2.6) -- (15.3,2.9);
\draw[fill,color=pink] (13.5,4.4) rectangle (15,4.6);
\draw[fill,color=pink] (13.5,4.6) -- (13.8,4.9) -- (15.3,4.9) -- (15,4.6) -- (13.5,4.6);
\draw[fill,color=pink] (15.3,4.9) -- (15.3,4.05) -- (15,4.05) -- (15,4.9) -- (15.3,4.9);
\draw[fill,color=pink] (11.5,2.4) rectangle (12,2.8);
\draw[fill,color=pink] (11.5,2.8) -- (11.7,3) -- (12,3) -- (12,2.8) -- (11.5,2.8);
\draw[fill,color=pink] (11.5,4.4) rectangle (12,4.8);
\draw[fill,color=pink] (11.5,4.8) -- (11.7,5) -- (12,5) -- (12,4.8) -- (11.5,4.8);
\draw (9,0) rectangle (14.7,5.7);
\draw (9,5.7) -- (11,7.7) -- (16.7,7.7) -- (14.7,5.7);
\draw (14.7,0) -- (16.7,2) -- (16.7,7.7) -- (14.7,5.7);
\draw[fill,color=blue] (10,6) circle (0.06cm);
\node at (9.8,5.8) {\tiny{$v_2$}};
\end{tikzpicture}
\caption{Forbidden region (the union of all boxes) and state space (its complement): One resource r of capacity 2 (left), resp.\ 1 (right); three processors each executing the program $PrVrPrVr$, cf Example \ref{ex:cons}(2).}\label{fig:PV}
\end{figure}

The common boundary of $F(r)$ and of $X(r)$ consists of those $\mb{x}\in\prod_{j=1}^n[0, l(k_j)+1]$ satisfying
\begin{itemize}
\item $cr(\mb{x})\le\kappa (r)$, and
\item $\exists \mb{t}=(t_1,\dots ,t_n), -1<t_j<1, t_j=0$ for $x_j\not\in\mb{Z}:\; cr(\mb{x}+\mb{t})>\kappa (r).$
\end{itemize}
In particular, a \emph{vertex} $v=(i_1,\dots ,i_n)\in X(r)_0$ with integer coordinates is contained in this boundary if and only if
\begin{itemize}
\item $cr(v)\le\kappa (r)$
\item $cr(v)+|\{ j|\; i_j\in Pr_j([0:k_j(r)])\cup Vr_j([0:k_j(r)])\}|>\kappa (r)$.
\end{itemize}
The first condition ensures that $v\in X(r)_0$. The second condition has as consequence that a point $v+(t_1,\dots ,t_n)$ -- with $0<t_j<1$ for $i_j\in Pr_j([0:k_j(r)])$ and $-1<t_j<0$ for $i_j\in Vr_j([0:k_j(r)])$ and $t_j=0$ else -- is contained in $F(r)$.

\subsubsection{Several resources}\label{sss:several2}
The forbidden region $F:=F(R)$ corresponding to a \emph{set} $R$ of resources is the \emph{union} $F:=\bigcup F(r)$ of the forbidden regions $F(r),\; r\in R$. It agrees with\\ $\{\mb{x}=(x_1,\dots ,x_n)\in\prod_{j=1}^n[0, l(k_j)+1]|\; \exists r\in R: cr(\mb{x})>\kappa (r)\}$.\\ The state space is its complement: $X=X(R):=\prod_{j=1}^n[0, l(k_j)+1]\setminus F=\bigcap_{r\in R}X(r)$.\\ It agrees with $\{\mb{x}=(x_1,\dots ,x_n)\in\prod_{j=1}^n[0, l(k_j)+1]|\; \forall r\in R: cr(x)\le\kappa (r)\}$.

A vertex $v=(i_1,\dots ,i_n)\in X(R)_0$ with integer coordinates in the common boundary $\partial F=\partial X$ is characterized by the following properties:
\begin{enumerate}
\item $cr(v)\le\kappa (r)$ \emph{for all} $r\in R$;
\item $\exists r\in R: cr(v)+|\{ j|\; i_j\in Pr_j([0:k_j(r)])\cup Vr_j([0:k_j(r)])\}|>\kappa (r)$.
\end{enumerate}

In the following, we focus on properties of spaces $\vec{P}(X)_s^t$ of directed paths (cf Section \ref{ss:Euclid}) in state space $X$  between two vertices (in $X_0$) with integer coordinates. 
\begin{remark}
Higher Dimensional Automata, as mentioned in the abstract, are far reaching generalizations of the state spaces corresponding to Dijkstra's concurrent $PV$ programs. Introduced by Pratt and van Glabbeek (cf \cite{Pratt, Glabbeek:05}) as generalizations of (labelled) transition systems, their underlying geometry is that of a pre-cubical complex, a glueing of directed cubes of various dimensions; not necessarily embedded in a cubular tiling of a Euclidean space. Consult \cite{FGHMR:16,Raussen:10,Raussen:12,Ziemianski:17,Ziemianski:20} for definitions and for combinatorial/topological descriptions of the spaces of directed paths from a source to a target. 
\end{remark}

\section{Future links and their connectivity}
\subsection{Future links}\label{ss:fl}
Let $I^n=[0,1]$ denote the standard $n$-cube. Any subcube (aka face) of $I^n$ containing the minimal vertex $\mb{0}$ is characterized by the set of 1-coordinates of its \emph{maximal} vertex. The poset of the subcubes properly containing $\mb{0}$ as minimal vertex is thus in an order-preserving correspondence with the non-empty subsets of $[1:n]$ and forms an $(n-1)$ dimensional ``future'' simplex $\Delta_{n-1}$. For a cubical \emph{sub}complex $X\subseteq I^n$ containing $\mb{0}$, consider the \emph{future link} $lk^+(X,\mb{0})\subset\Delta_{n-1}$ in $X$ consisting of those simplices corresponding to the subcubes \emph{contained in} $X$. 

The \emph{future link} $lk^+(X,v)$ of a vertex $v=(v_1,\dots ,v_n)$ (with integer coordinates) in a Euclidean cubical complex $X\subset\mb{R}^n$ is similarly defined encoding the subcomplex $X\cap\prod_1^n [v_j,v_j+1]$ contained in the unit cube ``over'' $v$. By definition, it is a simplicial complex embedded in the future simplex $\Delta_{n-1}(v)$ describing all cubes in $\mb{R}^n$ properly containing $v$ as minimal vertex; cf Ziemia\'{n}ski \cite[Def.\ 5.1]{Ziemianski:16} and Belton etal \cite[Def.\ 4]{Beltonetal:19}, \cite[Sect.\ 2.2]{Beltonetal:21} who deal with analogously defined past links. Every directed path starting at $v$ proceeds, for a while, in a face of the future link $lk^+(X(r),v)$. The future link thus tells us about the possible directions such a path can (or cannot) take; locally.

\subsection{A single resource}\label{ss:one}
Let $v=(v_1,\dots ,v_n)\in X(r)\subseteq \prod_{j=1}^n[0, l(k_j)+1]$ denote a vertex in the state space $X(r)$ corresponding to a single resource $r$; cf Section \ref{sss:one}. Let $m(v)$ denote the number of coordinates that are \emph{not} maximal (ie $v_j\neq l(k_j)+1$). To the maximal cube $M(v)\subset\prod_{j=1}^n[0, l(k_j)+1]$ with lower vertex $v$ -- of dimension $m(v)$ -- corresponds its future simplex $\Delta^{m(v)-1}(v)$. The future link $lk^+(X(r),v)\subseteq\Delta^{m(v)-1}(v)$ encodes those faces of $M(v)$ that are contained in $X(r)$. 

For a single resource $r$, the consumption function takes value $0$ at every vertex $v\in X_0$. Hence, the capacity $\kappa (r)$ of resource $r$ has the following consequence for future links:

\begin{lemma}\label{lem:link1}
The future link $lk^+(X(r),v)$ at a vertex $v=(v_1,\dots ,v_n)$
\begin{enumerate}
\item is the $(\kappa (r)-1)$-skeleton $\Delta^{m(v)-1}_{(\kappa (r)-1)}(v)$ of the future simplex $\Delta^{m(v)-1}(v)$ if all coordinates $v_j$ are of the form $Pr_j(i)$ or the final $l(k_j)+1$. 
\item is contractible otherwise. 
\end{enumerate}
\end{lemma}

\begin{proof}
\begin{enumerate}
\item The maximal subcubes in $X(r)$ containing $v$ as minimal vertex are in one-to-one correspondence with subsets of cardinality $\kappa (r)$ among the coordinates in $v$ that are not final.
\item Suppose that the coordinate $v_j$ is neither of the form $Pr_j(i)$ nor the final $l(k_j)+1$. Consider the subsimplex $\Delta_j^{m(v)-2}(v)\subset\Delta^{m(v)-1}(v) $ consisting of all subsets \emph{not containing} $j$. For every subcube
$Q\subset X(r)$ with minimal vertex $v$ and maximal vertex $v'$ with \emph{same} $j$-th coordinate $v_j$ (with associated simplex contained in $\Delta_j^{m(v)-2}(v)$), the product $Q\times I$ (in direction $j$, with associated simplex contained in $\Delta^{m(v)-1}(v)$) is also contained in $X(r)$. In other words, the future link $lk^+(X(r),v)$ is a \emph{cone} with apex corresponding to $j$; and hence contractible.
\end{enumerate}
\end{proof}

\begin{corollary}
The future link $lk^+(X(r),v)$ at a vertex $v=(v_1,\dots ,v_n)\in X_0$ is
\begin{enumerate}
\item $(\kappa (r)-2)$-connected but not $(\kappa (r)-1)$-connected if all coordinates $v_j$ are of the form $Pr_j(i)$ or the final $l(k_j)+1$ and if $m(v)>\kappa(r)$.
\item contractible if $m(v)\le \kappa (r)$ or if $v$ has at least one coordinate $v_j$ not of the form $Pr_j(i)$ or the final $l(k_j)+1$.
\end{enumerate}
\end{corollary}

\begin{example}\label{ex:lk+}
We refer to Figure \ref{fig:PV} and Figure \ref{fig:lk+conn}. The future link $lk^+(X(r);v_0)$ of the vertex $v_0=(1,1,1)$ on the left hand side of Figure \ref{fig:PV} is the $1$-skeleton ($\kappa (r)-1=1$) of the future simplex $\Delta^2(v_0)$; a connected, but not simply-connected hollow triangle. Interpretation: Two out of three processors can proceed at $v_0$.\\
The three remaining vertices occur when $\kappa (r)-1=0$ (right hand side of Figure \ref{fig:PV}). The future link $lk^+(X(r);v_1)$ at $v_1=(1,1,0)$ is the cone over two points representing the future link of $v_1$ \emph{restricted to the bottom plane}; it is contractible. Interpretation: Processor 3 can proceed independently of the two others.\\
The future link $lk^+(X(r);v_2)$ of the vertex $v_2=(1,1,4)$ is the $0$-skeleton of the future simplex $\Delta^1(v_2)$ -- in the top plane, consisting of two points. The drawing on the right represents a cutout of the forbidden region from Figure \ref{fig:PV}) close to the vertex $v_3=(1,1,1)$ (not visible in that figure). Among the faces of the future cube $M(v_3)$ (spanned by $(1,1,1)$ and $(2,2,2)$) adjacent to $v_3$, only the three dashed edges belong to the state space $X(r)$; they give rise to the $\kappa (r)-1=0$-skeleton of the future simplex $\Delta^2(v_3)$ (represented by three points). Interpretation: At $v_2$ and at $v_3$, only one of the two, resp.\ three active processors may proceed.

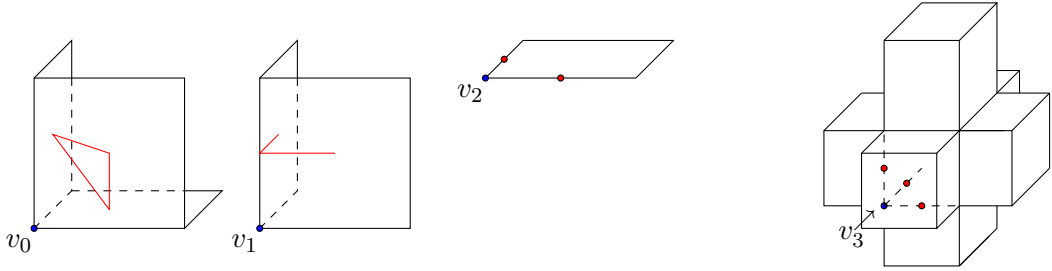
\begin{figure}[h]
\begin{tikzpicture}
\draw (0,0) rectangle (2,2);
\draw (0,2) -- (0.5,2.5) -- (0.5,2);
\draw [dashed] (0.5,2) -- (0.5,0.5) -- (0,0);
\draw [dashed] (0.5,0.5) -- (2,0.5);
\draw (2,0.5) -- (2.5,0.5) -- (2,0);
\draw [color=red] (1,0.25) -- (0.25,1.25) -- (1,1) -- (1,0.25);
\draw (3,0) rectangle (5,2);
\draw (3,2) -- (3.5,2.5) -- (3.5,2);
\draw [dashed] (3.5,2) -- (3.5,0.5) -- (3,0);
\draw [color=red] (4,1) -- (3,1) -- (3.25,1.25);
\draw (6,2) -- (6.5,2.5) -- (8.5,2.5) -- (8,2) -- (6,2);
\draw [fill=red] (7,2) circle (0.04cm);
\draw [fill=red] (6.25,2.25) circle (0.04cm);
\draw [fill=blue] (0,0) circle (0.04cm);
\draw [fill=blue] (3,0) circle (0.04cm);
\draw [fill=blue] (6,2) circle (0.04cm);
\node at (-0.2,-0.2) {$v_0$};
\node at (2.8,-0.2) {$v_1$};
\node at (5.8,1.8) {$v_2$};
\draw (11,0) rectangle (12,1);
\draw (11,1) -- (11.3,1.3);
\draw (12,0) -- (12.3,0.3); 
\draw (12,1) -- (13.1,2.1);
\draw (11,0.3) -- (10.5,0.3) -- (10.5,1.3) -- (12.9,1.3);
\draw (12.3,1.3) -- (13,1.3) -- (13,0.3) -- (12.3,0.3);
\draw (11,1.8) -- (11.3,1.8);
\draw (12.8,1.8) -- (13.5,1.8);
\draw (13.5,0.8) -- (13.5,1.8);
\draw (10.5,1.3) -- (11,1.8);
\draw (13,0.3) -- (13.5,0.8);
\draw (13,1.3) -- (13.5,1.8);
\draw [fill=blue] (11.3,0.3) circle (0.04cm);
\draw (11.3,1.3) -- (11.3,2.5);
\draw (12.3,2.5) -- (12.3,-0.5);
\draw (11.3,2.5) -- (11.8,3) -- (12.8,3) -- (12.3,2.5) -- (11.3,2.5);
\draw (12.8,3) -- (12.8,1.8);
\draw (12.8, 0.3) -- (12.8,0) -- (12.3,-0.5) -- (11.3,-0.5) -- (11.3,0);
\draw (12.3,-0.5) -- (12.7,-0.1);
\draw [dashed] (11.3,0.3) -- (12.3,0.3);
\draw [dashed] (11.3,0.3) -- (11.3,1.3);
\draw [dashed] (11.3,0.3) -- (11.8,0.8);
\draw (12.8,2.1) -- (13.1,2.1) -- (13.1,1.8);
\node at (10.95,0) {$v_3^{\nearrow}$};
\draw [fill=red] (11.3,0.8) circle (0.04cm);
\draw [fill=red] (11.8,0.3) circle (0.04cm);
\draw [fill=red] (11.6,0.6) circle (0.04cm);
\end{tikzpicture}
\caption{Future links of four vertices in Figure \ref{fig:PV}.}
\label{fig:lk+conn}
\end{figure}
\end{example}


\begin{remark} 
For a single resource $r$ with given capacity $\kappa$, information about the homology, ie the Betti numbers, of the path space $\vec{P}(X(r))_{\mb{0}}^{2\mb{k}+\mb{1}},\; \mb{k}=(k_1,\dots ,k_n)\in \mb{N}_{\ge 0}^n$, can be found in Meshulam-Raussen \cite[Corollary 5.2]{MR:17}.
\end{remark}

\subsection{Several shared resources}\label{ss:sev}

Let us now consider a program on the resource set $R$ and a vertex $v=(v_1,\dots ,v_n)\in X_0$ in the corresponding state space $X:=X(R)$, cf Section \ref{sss:several}. Every coordinate $v_j$ is either $0$ or maximal or in the range of one of the functions $Pr$, resp.\ $Vr,\; r\in R$. It is only relevant to investigate the future links of vertices $v$ with \emph{all} coordinates either maximal or in the range of a function $Pr_j,\; r\in R, j\in P$, since
\begin{lemma}\label{lem:contr}
Let one of the coordinates $v_j$ be either $0$ or of the form $Vr(i_j)$ for some resource $r\in R$. Then the future $lk^+(X,v)$ is contractible.
\end{lemma}
\begin{proof}
As in the proof of Lemma \ref{lem:link1}, the future link is then a cone and hence contractible.
\end{proof}

We need some notation to explore future links of the remaining vertices. 

\begin{definition}\label{def:X0P}
We let $X_0^P$ denote the set of those vertices $v\in X_0$ with all coordinates maximal or in the range of one of the $Pr$-functions. For such a vertex $v$:
\begin{itemize}
\item $P(v)\subseteq P$ denotes the subset of \emph{active} processors $j$, with $v_j\neq l(k_j)+1$ (not a final state); with cardinality $m(v):=|P(v)|$.
\item $R(v)\subseteq R$ denotes the subset of \emph{resources} $r$ such that there exists a $j\in [1:n]$ with $v_j$ in the range of $Pr$ (a lock to $r$ is requested at $v$).
\item For $r\in R(v)$, we let $Pr(v)$ denote $\{ j\in P(v)| v_j\in Pr([1:k_j(r)])\}$ the set of processors with a call to resource $r$ at $v$. Its cardinality $dr(v):=|Pr(v)|>0$ (cf Section \ref{sss:several}) is the number of calls to $r$ issued at $v$. The sum $\sum_{r\in R(v)}dr(v)$ is equal to $m(v)=|P(v)|$ -- if every processor calls exactly one resource at a ``time''.
\end{itemize}
\end{definition}

With this in place, we can formulate the following crucial simple technical result determining the topology of a future link $lk^+(X;v)$ at a vertex $v\in X_0^P$. Remember that $cr(v)$ denotes the consumption of locks to resource $r$ at the vertex $v$ (cf Section \ref{sss:several}).

\begin{proposition}\label{prop:linkjoin}
For $v\in X_0^P$, the future link $lk^+(X,v)$ is homeomorphic to a join \emph{(}consisting of all convex combinations; notation $*$\emph{)} of skeleton spaces \[ lk^+(X,v)\cong\Coast_{r\in R(v)}\Delta^{dr(v)-1}_{(\kappa(r)-cr(v)-1)}=\Coast_{r\in R(v),\; cr(v)<\kappa (r)}\Delta^{dr(v)-1}_{(\kappa(r)-cr(v)-1)}.\]
\end{proposition}

Join with an empty set -- occurring in the left hand join if $cr(v)=\kappa (r)$ -- has to be understood as $A*\emptyset =A$.
\begin{proof}
Decompose the maximal cube $M(v)\subset\prod_{j=1}^n[0, l(k_j)+1]$ with minimal vertex $v\in X_0^P$ -- of dimension $m(v)$ -- as a product $\prod_{r\in R(v)}M(r)$ of cubes $M(r)$ in directions $j\in Pr(v)$; each of dimension $dr(v), r\in R(v)$. A (future) subcube $C\subseteq X\cap M(v)$ decomposes correspondingly as $C=\prod_{r\in R(v)}C(r)$. For a future cube $C$ to be contained in $X(R)$, the capacity constraint regarding $r\in R$ requires  precisely that $\dim C(r)\le \kappa (r)-cr(v)$: Since $cr(v)$ locks are already active at $v$, only up to $\kappa (r)-cr(v)$ new locks can be acquired at $v$. A directed path can increase in up to $\kappa (r)-cr(v)$ directions $j\in Pr(v)$ from $v$.

Products of cubes correspond to \emph{unions} of sets of indices, and these correspond to \emph{joins} of the associated future links of $v$ with respect to each of the $C(r),\; r\in R(v)$. Apply Lemma \ref{lem:link1}.
\end{proof}

\begin{remark}
Proposition \ref{prop:linkjoin} is only true -- on the nose -- if one assumes that every processor makes call to resources subsequently, not at the same ``time''. If this is not the case, one may split up two concurrent calls without changing the homotopy type of the path space, as observed by Ziemia\'{n}ski \cite[Sect.\ 3]{Ziemianski:16}.
\end{remark}

\subsection{Spare capacity at vertices}
Our next aim is to determine the connectivity of future links. When are they path-connected or of higher connectivity? To answer that question, we introduce the notion of spare capacity that can be determined directly from the semantics of a given PV program; cf Section \ref{ss:sparecapalg} for details.

\begin{definition}\label{def:cc}
For every vertex $v\in X_0^P$, let $\chi_v: R\to \{ 0,1\}$ denote the characteristic function of $R(v)\subset R$.

The \emph{spare capacity} $\kappa(X;v)$ of the state space $X$ at the (allowed) vertex $v\in X_0^P$ is defined as \[\kappa(X;v):=\begin{cases}\infty & dr(v)\le\kappa (r)-cr(v)>0 \mbox{ for at least one } r\in R(v)\\ \sum_{r\in R}\chi_v(r)(\kappa(r)-cr(v)) & \mbox{ else}\end{cases}.\]
\end{definition}
If finite, the spare capacity at $v$ measures the maximal dimension of a subcube of the cube with bottom vertex $v$ that is contained in the state space $X$, ie the maximal number of processors that can take a step at $v$ simultaneously. This will be explained in detail in Proposition \ref{cor:linkconn}.

\begin{remark}\label{rem:dead}
Inequalities (1) and (2) at the end of Section \ref{sss:several} have the following consequences:
\begin{enumerate}
\item If a vertex $v$ is contained in the state space $X$, then $cr(v)\le \kappa (r)$ for each $r\in R(v)$. Hence $\kappa(X;v)\ge 0$ for every $v\in X_0^P$, cf Section \ref{ss:forbidden}.
\item For a vertex $v\in X_0^P$ in the (lower) \emph{boundary} of the state space $X$, we have for every $r\in R(v)$ moreover: $dr(v)>\kappa (r)-cr(v)$, cf Section \ref{ss:forbidden}. In particular, $\kappa(X;v)$ is finite for such a vertex. 
\item In the opposite direction, assume $v\in X_0^P$ with finite spare capacity $\kappa(X;v)$. For every resource $r\in R(v)$ this requires $cr(v)=\kappa (r)$ or $dr(v)>\kappa (r)-cr(v)$.  Since $dr(v)>0$ for every $r\in R(v)$, the inequality $dr(v)>\kappa (r)-cr(v)$ holds in both cases. Hence $v$ is contained in the \emph{boundary} of $X$; cf (2) at the end of Section \ref{sss:several}. 
\item For $\kappa(X;v)<\infty$, the definition of spare capacity is equivalent to\\ $\kappa(X;v)= \sum_{r\in R(v)} \kappa(r)-\sum_{r\in R(v)}cr (v)$.
\end{enumerate}
\end{remark}

\begin{example}
Let us determine spare capacities and future links at vertices $v\in X_0^P$ in the boundary of the forbidden region in the three cases from Example \ref{ex:cons}:
\begin{enumerate} 
\item Example \ref{ex:cons}(1):
The relevant vertices in $X_0^P$ are $(2,1), (1,2)$ and $(2,2)$, cf Figure \ref{fig:PV}. Using the formula in Remark \ref{rem:dead}(4), we obtain $\kappa (X;(2,1))=\kappa (X;(1,2))=2-1=1$, whereas $\kappa (X;(2,2))=2-2=0$.\\ The future link $lk^+(X;v)$ consists of two disjoint points in the first two cases (the $0$-skeleton of an edge), and it is empty in the last.
\item Example \ref{ex:cons}(2), cf Figure \ref{fig:PV} and Figure \ref{fig:lk+conn}: For $\kappa (r)=2$ and the vertex $v_0=(1,1,1)$, we obtain $\kappa (X;v_0)=2-0=2$.\\
For $\kappa (r)=1$, the spare capacities $\kappa (X;v_i)$ at the vertices $v_2=(1,1,4)$ and $v_3=(1,1,1)$ are both equal to $1-0=1$. The vertex $v_1=(1,1,0)$ is not contained in $X_0^P$.\\
Remark that these results are consistent with the determination of future links in  Example \ref{ex:lk+}.
\item Example \ref{ex:cons}(3):
\begin{description} 
\item[A] At the vertex $v=(2,2,2,2)\in X_0^P$ with the final lock requests, $R(v)=R, \kappa (X;v)$ $=3+3-3-1=2$ and $lk^+(X;v)=\Delta ^2_{(1)}$, the 1-skeleton of a 2-simplex homeomorphic to $S^1$ (path-connected, but not simply-connected). Interpretation: $4$ cannot proceed whereas two of the processors $i,\; i\le 3,$ can proceed concurrently.   
\item[B] At the vertex $v=(2,2,2,2)$, $\kappa (X;v)=3+3-2-2=2$ and $lk^+(X;v)=\Delta^1_{(0)}\ast\Delta^1_{(0)}$, the join of two spaces consisting of two points each; the resulting graph (four edges and four vertices) is homotopy equivalent to $S^1$.\\ Interpretation: One of processors $1,2$ and also one of $3,4$ can proceed simultaneously; the local future at $v$ consists of four (2-dimensional) square faces, corresponding to the four choices of two pairs.
\end{description}
\item Example \ref{ex:cons}(4):
\begin{description}
\item[A] The vertex $w=(2,2,2,2,2)$ is no longer reachable: $cr(w)=4>3=\kappa (r)$.
\item[B] In this case, $\kappa(X;w)=3+3-3-2=1$, and $lk^+(X,v)=\Delta^2_{(0)}$ is a (non-connected) 3 point space.\\  Interpretation: Only one of processors $0,1,2$ can proceed along an edge; processors $3$ and $4$ have to wait. 
\end{description}
\end{enumerate}
\end{example}

More systematically, the spare capacity $\kappa (X;v)$ allows to determine the connectivity of the future link $lk^+(X,v)$:

\begin{proposition}\label{cor:linkconn}
Let $v\in X_0^P$ denote a vertex in the state space $X$. 
\begin{enumerate}
\item If $\kappa(X;v)=\infty$, then $lk^+(X,v)$ is contractible.
\item If $\kappa(X;v)<\infty$, then $\kappa(X;v)\le n-|R(v)|$.
\item $v$ is a deadlock vertex (ie the only directed path starting at $v$ is the constant path with value $v$ or equivalently $lk^+(X,v)=\emptyset$) if and only if $\kappa(X;v)=0$. 
\item The future link $lk^+(X,v)$ is disconnected if and only if $\kappa(X;v)=1$.
\item The future link $lk^+(X,v)$ is path-connected but not simply connected if and only if $\kappa(X;v)=2$. 
\item If $2<\kappa(X;v)<\infty$, then $lk^+(X,v)$ is $(\kappa(X;v)-2)$-connected but not $(\kappa(X;v)-1)$-connected.
\end{enumerate}
\end{proposition}

\begin{proof}
\begin{enumerate}
\item If $dr(v)\le\kappa (r)-cr(v)$ for some $r\in R(v)$, then the skeleton corresponding to $r$ in the join decomposition from Proposition \ref{prop:linkjoin} is the \emph{entire} simplex $\Delta^{dr(v)-1}(v)$ (non-empty since $dr(v)>0$) and hence contractible: All processors locking $r$ at $v$ can proceed independently of each other. Moreover, a join with a contractible space (in Proposition \ref{prop:linkjoin}) is contractible.
\item For each $r\in R(v)$ we have that $cr(v)-\kappa (r)<dr(v)$ -- even if the left hand side is $0$. Since we are dealing with integers, this means that $cr(v)-\kappa (r)\le dr(v)- 1$. Summing up over all $r\in R(v)$, we get the spare capacity on the left and $n-|R(v)|$ on the right.
\item Spare capacity $\kappa(X;v)=0$ for $v\in X_0$ if and only if $cr(v)=\kappa (r)$ for every $r\in R(v)$. Hence, none of the processors can acquire an additional lock at $v$; or more technically, all skeleta $\Delta^{dr(v)-1}_{(\kappa(r)-cr(v)-1)},\; r\in R(v)$, are empty, and so is $lk^+(X;v)$.\\ If $\kappa(X;v)>0$, then $cr(v)<\kappa (r)$ for some $r\in R(v)$. At least one processor can proceed from $v$; at least one of the skeleta is non-empty, and so is then $lk^+(X;v)$.
\item The spare capacity $\kappa(X;v)=1$ if and only if $\kappa (r)=cr(v)$ for every $r\in R(v)$ apart from a single resource $r_0\in R(v)$ with $\kappa (r_0)=cr_0(v)+1$. In that case $lk^+(X,v)$ is the $0$-skeleton of $\Delta^{dr_0(v)-1}$ which is not path-connected since $dr_0(v)>\kappa (r_0)-cr_0(v)=1$.
\item The spare capacity $\kappa(X;v)=2$ if and only if either there is a single resource $r_0\in R(v)$ with $\kappa (r_0)=cr_0(v)+2$ and $\kappa (r)=cr(v)$ for all others or if there are two resources $r_1, r_2\in R(v)$ with $\kappa (r_i)=cr_i(v)+1$ and $\kappa (r)=cr(v)$ for all others.\\ In the first case, $lk^+(X,v)$ is the $1$-skeleton of a simplex $\Delta^{dr_0(v)-1}$, which is path-connected, but not simply-connected, since $dr_0(v)>\kappa (r_0)-cr_0(v)=2$.\\ In the second case $lk^+(X,v)$ is the join of two discrete spaces which is path-connected. The future link cannot be simply-connected:  Each of the discrete spaces contains at least two points since $dr_i(v)>\kappa (r_i)-cr_i(v)=1$.
\item If $\kappa(X;v)\ge 3$, then $lk^+(X,v)$ is either the $(\kappa(X;v)-1)$-skeleton of a non-empty simplex or the join of a path-connected space with another space, and therefore simply-connected; connectivity can thus be read off from homology.\\ The Mayer-Vietoris sequence in homology (cf.\ \cite[ch.\ 2.2]{Hatcher:02}) applied to a join $A*B=CA\times B\cup_{A\times B}A\times CB$ shows: If $A$ $k$-connected and $B$ is $l$-connected, then $A*B$ is $(k+l+1)$-connected. Inductively, this implies that the join in Proposition \ref{prop:linkjoin} is trivial in dimensions up to $\kappa(X;v)-2$ and non-trivial in dimension $\kappa(X;v)-1$. The Betti number in that dimension is the sum of the top-dimensional Betti numbers of the skeleta of the simplices involved. 
\end{enumerate}
\end{proof}

\begin{remark}\label{rem:notpc}
Let us examine (4) in Proposition \ref{cor:linkconn} more closely: A vertex $v\in X_0^P$ has a future link $lk^+(X;v)$ that is \emph{not path-connected} if and only if $\kappa(X;v)=1$. This is the case if all resources $r\in R(v)$ are exhausted ($\kappa (r)=cr(v)$) \emph{except} for a single resource $r_0\in R(v)$ with $dr_0(v)>\kappa (r_0)-cr_0(v)=1$.\\
We will call such a vertex $v$ with $\kappa (X;v)=1$ a \emph{critical} vertex (or state).
\end{remark}

\begin{remark}
Is it necessary to analyse the local future of (points on) faces of the state spaces, as well? This is not the case as long as we only consider spaces of directed paths whose target is a \emph{vertex} $t$. It is shown in Fajstrup \cite{Fajstrup:05} that 
$\vec{P}(X)_v^{t}$ and $\vec{P}(X)_{v_0}^{t}$ are homotopy equivalent if $v$ is a point on a face with top vertex $v_0$. 
\end{remark}

\section{Global connectivity of spaces of directed paths}
\subsection{Spare capacity of a concurrent PV program}
We fix a PV-program with state space $X$ and target vertex $t$. All future links have to be understood with respect to that target $t$. Our aim is to establish connectivity bounds for spaces $\vec{P}(X)_v^t$ of directed paths within $X$ starting at a vertex $v\in X_0$ and ending at $t$, endowed with the compact-open (aka uniform convergence) topology; in particular, to find out whether such spaces are path-connected, via \emph{directed} homotopies, ie via 1-parameter families of \emph{directed} paths; cf Fajstrup etal \cite[ch.\ 4.2]{FGHMR:16}.

\begin{definition}\label{def:cumcap}
The \emph{spare capacity} $\kappa(X)$ of a concurrent $PV$ program with state space $X$ is defined as $\kappa(X):=\min_{v\in X_0^P, v\le t}\kappa(X;v)$, ie the \emph{minimum} of all spare capacities $\kappa(X;v)$ of (allowed) vertices $v\in X_0^P$ from which $t$ is reachable.
\end{definition}

\begin{remark}
\begin{enumerate}
\item We can restrict attention to vertices $v\in X_0^P$ (Definition \ref{def:X0P}) since the links of all other vertices are contractible (Lemma \ref{lem:contr}).
\item It is possible to remove the doomed region (cf \cite[ch.\ 5]{FGHMR:16}); the set of all elements that cannot be connected to  $t$ by a directed path) from the state space in a first algorithmic step. Only the spare capacities of vertices in the new smaller state space have to be estimated; cf Section \ref{ss:deaddoom} for details. 
\end{enumerate}
\end{remark}

\subsection{The spare capacity as a connectivity indicator}\label{ss:sparecap}
The theorem below formalizes the slogan: If there are no local obstructions to (higher) connectivities of paths spaces, then there are no global ones, either:
\begin{theorem}\label{prop:PVconn}
Given the state space $X$ of a concurrent \emph{(}linear\emph{)} $PV$-program with final state $t$ and a vertex $v_0\in X_0^P$ from which $t$ is reachable. Then the path space $\vec{P}(X)_{v_0}^t$ is $(\kappa(X)-2)$-connected. 
\end{theorem}

\begin{proof}
The proof makes use of the future link $lk^+(X,v_0)$ of a vertex $v_0\in X_0$ viewed as a \emph{poset category}, with inclusion (of index sets; or simplices) as partial order. It proceeds by induction on the $L_1$ (aka taxicab) distance between vertices $v_0$ and $t$ (a non-negative integer!) in the cubical grid; cf Raussen \cite[Sect.\ 2.2]{Raussen:09} in this context. It starts with distance $0$, ie $v_0=t$. Then the path space consists of the constant path only, and it is hence contractible.

Assume by induction that $\vec{P}(X)_v^t$ is $(\kappa(X)-2)$-connected for every vertex $v$ satisfying $v_0<v\le t$, ie $v_0\neq v$ and there exist directed paths from $v_0$ to $v$ and from $v$ to $t$ in $X$.   In the following, we restrict attention to those vertices $v_0<v$ giving rise to a simplex in $lk^+(X;v_0)$, cf Section \ref{ss:fl}; abusing notation, we write $v\in lk^+(X;v_0)$ for these vertices. It is shown in Raussen-Ziemia\'{n}ski \cite[Sect.\ 2.3]{RZ:14}, formulated for past links, and subsequently exploited in Ziemia\'{n}ski \cite{Ziemianski:16} and Belton etal \cite{Beltonetal:19,Beltonetal:21}, that $\vec{P}(X)_{v_0}^t$ is the \emph{colimit} (the union) of certain subspaces $F_v\vec{P}(X)_{v_0}^t, v\in lk^+(X,v_0),$ such that $F_v\vec{P}(X)_{v_0}^t$ is \emph{homotopy equivalent to} $\vec{P}(X)_v^t$. By the induction hypothesis, $F_v\vec{P}(X)_{v_0}^t$ is thus $(\kappa(X)-2)$-connected for each vertex $v$ in the future link $lk^+(X,v_0)$.

The geometric realization of the future link category (aka its nerve) is the future link space $lk^+(X,v_0)$ that is (at least) $(\kappa(X)-2)$-connected by Corollary \ref{cor:linkconn}. 
We apply Bj\"{o}rner's theorem \cite[Theorem 6]{Bjoerner:03} on a colimit of spaces whose connectivity is limited below in a certain pattern. These connectivity conditions are certainly (more than) met. Hence the colimit of the spaces $F_v\vec{P}(X)_{v_0}^t$, and therefore the path space $\vec{P}(X)_{v_0}^t$, has homotopy groups isomorphic to those of the nerve of the (future link) category $lk^+(X;v_0)$ up to dimension $\kappa(X)-2$. In particular, the path space $\vec{P}(X)_{v_0}^t$ is also $(\kappa(X)-2)$-connected (even if $v\not\in X_0^P$; then the future link is contractible by Lemma \ref{lem:contr}). 

Bj\"{o}rner's theorem is formulated for simplicial complexes and their subcomplexes. This restriction is of no concern for us since all occurring path spaces have the homotopy type of CW-complexes, cf Raussen \cite[Prop.\ 3.15]{Raussen:09}, and hence of simplicial complexes, cf Hatcher \cite[Theorem 2C.5]{Hatcher:02}. 
\end{proof}

As a special case, we obtain
\begin{corollary}\label{cor:conn}
Let $X$ denote a Euclidean cubical complex with vertices $v_0\le t\in X_0.$\\
If $\kappa (X)\ge 1$, then $\vec{P}(X)_{v_0}^t$ is non-empty.
If $\kappa(X)\ge 2$, then $\vec{P}(X)_{v_0}^t$ is path-connected. 
\end{corollary} 
\begin{remark}
\begin{enumerate}
\item Corollary \ref{cor:conn} has the following interpretation:\\
If $\kappa (X)\ge 1$, then the state space $X$ does not contain any deadlock (cf Proposition \ref{cor:linkconn}(3)), and every execution at a vertex $v_0\le t$ can terminate at $t$.\\
If $\kappa (X)\ge 2$, then every concurrent execution of individual programs on processors $j\in P$ starting at $v_0$ and ending at $t$ yields the same result, regardless of the order of access to shared resources.
\item The condition $\kappa(X)\ge 2$ in Corollary \ref{cor:conn} is satisfied if and only if for \emph{each} vertex $v\in X_0^P$ one of the following conditions is met:
\begin{itemize}
\item $\exists r\in R(v)$ with $dr(v)\le\kappa (r)-cr(v)$ ($v$ is not contained in the boundary of the state space $X$);
\item $\exists r_1\neq r_2\in R(v)$ with $\kappa(r_i)\ge cr_i(v)+1$ (two resources with non-exhausted capacity at $v$);
\item $\exists r\in R(v)$ with $\kappa (r)\ge cr(v)+2$ (that resource can be accessed by two processors concurrently at $v$).
\end{itemize}
\end{enumerate}
\end{remark}
Proposition \ref{prop:PVconn} and Corollary \ref{cor:conn} are strict in the following sense:

\begin{proposition}\label{prop:mini}
Given the state space $X$ of a PV-program with final state $t$ and a vertex $v_0\in X_0^P$ from which $t$ is reachable. Assume that $\kappa(X;v_0)\le \kappa(X;v)$ for every vertex $v_0\le v\le t$. Then the path space $\vec{P}(X)_{v_0}^t$ is $(\kappa(X;v_0)-2)$-connected but not $(\kappa(X;v_0)-1)$-connected; cf Proposition \emph{\ref{cor:linkconn}}. 
\end{proposition}

\begin{proof}
As in the proof of Proposition \ref{prop:PVconn}, we may assume inductively that all path spaces $\vec{P}(X)_v^t$ from vertices $v\in lk^+(v_0)$ are at least ($\kappa(X;v_0)-2$)-connected. Hence,
using \cite[Theorem 6]{Bjoerner:03}, $\vec{P}(X)_{v_0}^t$ has the same \emph{non-trivial} homotopy in dimension $\kappa(X;v_0)-1$ as $lk^+(X;v_0)$.
\end{proof}

\begin{corollary}\label{cor:mini}
Given the state space $X$ of a concurrent program and a vertex $v_0\in X$ with \emph{minimal} spare capacity, ie $\kappa(X)=\kappa(X;v_0)$. Then the path space $\vec{P}(X)_{v_0}^t$ is ($\kappa(X)-2$)-connected but not ($\kappa(X)-1$)-connected.
\end{corollary}

\begin{corollary}
Given the state space $X$ of a concurrent program and a vertex $v_0\in X_0^P$ with spare capacity $\kappa (X)=\kappa(X;v_0)=1$. Then the path space $\vec{P}(X)_{v_0}^t$ is not path-connected, potentially giving rise to \emph{different} results of a concurrent computation depending on the order of access to shared resources.
\end{corollary}

See Figure \ref{fig:disconn} for an illustration. We have more to say on vertices $v\in X_0^P$ ``below $v_0$'' with non path-connected path spaces $\vec{P}(X)_v^t$ in Section \ref{sss:doomdis}.

\begin{center}
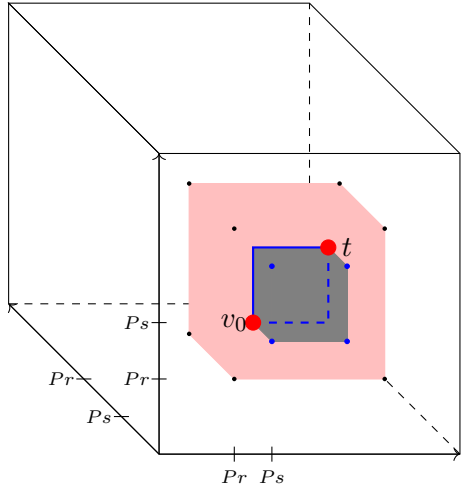
\begin{figure}[h]
\begin{tikzpicture}
\draw (0,0) rectangle (4,4);
\draw [->] (0,0) -- (4,0);
\draw [->] (0,0) -- (0,4);
\draw [->] (0,0) -- (-2,2);
\draw (1,-0.1) -- (1,0.1);
\draw (1.5,-0.1) -- (1.5,0.1);
\draw (-0.1,1) -- (0.1,1);
\draw (-0.1,1.75) -- (0.1,1.75);
\draw (-0.6,0.5) -- (-0.4,0.5);
\draw (-1.1,1) -- (-0.9,1);
\draw (0,4) -- (-2,6) -- (2,6) -- (4,4);
\draw (0,0) -- (-2,2) -- (-2,6) -- (0,4);
\draw [dashed] (4,0) -- (2,2) -- (2,6);
\draw [dashed] (-2,2) -- (2,2);
\draw [fill,color=pink] (1,1) rectangle (3,3);
\draw [fill,color=pink] (1,3) -- (0.4,3.6) -- (2.4,3.6) -- (3,3) -- (1,3);
\draw [fill,color=pink] (1,1) -- (1,3) -- (0.4,3.6) -- (0.4,1.6) -- (1,1);
\draw [fill, color=gray] (1.5,1.5) rectangle (2.5,2.5);
\draw [fill, color=gray] (1.5,2.5) -- (1.25,2.75) -- (2.25,2.75) -- (2.5,2.5) -- (1.5,2.5);
\draw [fill, color=gray] (1.5,1.5) -- (1.5,2.5) -- (1.25,2.75) -- (1.25,1.75) --  (1.5,1.5);
\draw [thick, color =blue] (1.25,1.75) -- (1.25,2.75) -- (2.25,2.75);
\draw [thick, color=blue,dashed] (1.25,1.75) -- (2.25,1.75) -- (2.25,2.75);  
\draw [fill, color=red] (1.25,1.75) circle (0.1cm);
\draw [fill, color=red] (2.25,2.75) circle (0.1cm);
\draw node at (1,1.75) {$v_0$};
\draw node at (2.5,2.75) {$t$};
\draw node at (1,-0.3) {\tiny{$Pr$}};
\draw node at (1.5,-0.3) {\tiny{$Ps$}};
\draw node at (-0.3,1) {\tiny{$Pr$}};
\draw node at (-0.3,1.75) {\tiny{$Ps$}};
\draw node at (-0.8,0.5) {\tiny{$Ps$}};
\draw node at (-1.3,1) {\tiny{$Pr$}};
\fill (1,1) circle (0.03cm);
\fill (1,3) circle (0.03cm);
\fill (3,1) circle (0.03cm);
\fill (3,3) circle (0.03cm);
\fill (0.4,3.6) circle (0.03cm);
\fill (2.4,3.6) circle (0.03cm);
\fill (0.4,1.6) circle (0.03cm);
\draw [fill,color=blue] (1.5,1.5) circle (0.03cm);
\draw [fill,color=blue] (2.5,2.5) circle (0.03cm);
\draw [fill,color=blue] (1.5,2.5) circle (0.03cm);
\draw [fill,color=blue] (2.5,1.5) circle (0.03cm);

\end{tikzpicture}
\caption{A disconnected path space arises for $n=3, R=\{ r, s\}, \kappa (r)=\kappa (s)=2$, and PV-programs $1,3$ starting both by $PrPs$ and $2$ by $PsPr$. The forbidden regions are indicated as $F(r)$ in pink, resp.\ $F(s)$ in grey; visible vertices are indicated by ticks. For the vertex $v_0=(Ps,Pr,Ps)$, we have $ds(v_0)=2, cs(v_0)=1$, and $dr(v_0)=1, cr(v_0)=2$. Hence $\kappa(X;v_0)=2+2-2-1=1$. Path space $\vec{P}(X)_{v_0}^t,\; t=(Vs,Pr,Vs)$ is in fact disconnected; likewise $\vec{P}(X)_{v_0}^{v_1}$  for every vertex $v_1$ with $t\le v_1$.}\label{fig:disconn}
\end{figure}
\end{center}

\begin{example}\label{ex:nonconn} Consider $l<n$ positive integers $\kappa_j<n, 1\le j\le l,$ such that $\sum_1^l\kappa_j>(l-1)n$. We construct a PV program on $n$ processors and $l$ resources $r_j$ with capacities $\kappa(r_j)=\kappa_j$ resulting in a corresponding state space $X$ with $\kappa (X)=\sum_1^l\kappa_j-(l-1)n$. There exist two vertices $v_0, v_1\in X$ such that $\vec{P}(X)_{v_0}^{v_1}$ has first non-trivial homology in dimension $\kappa(X)-1=\sum_1^l\kappa_j-(l-1)n-1$. In particular, $\vec{P}(X)_{v_0}^{v_1}$ is not path-connected for $\sum_1^l\kappa_j=(l-1)n+1$:

Define $\bar{c}_k=\sum_1^k\kappa_j-(k-1)n, \bar{d}_k=n-\bar{c}_k=kn-\sum_1^k\kappa_j$. Remark that \begin{equation}\label{eq:1}\bar{c}_k+\bar{d}_k=n\mbox{ and }\bar{c}_k+\bar{d}_{k-1}=\kappa_k.\end{equation} 
Consider a PV program where each of the $n$ processors individually executes a program of the form $P_{m_1}\dots P_{m_l}V_{m_l}\dots V_{m_1}$ with $(m_1,\dots ,m_l)$ a permutation in $\Sigma_l$, and $P_i$ is short for $Pr_i$. More specifically, they are chosen according to the following pattern (column $j$ corresponds to the $j$-th $P$ command in each thread; cf Figure \ref{fig:cap}): $P_1,\dots ,P_{l-1}$ are successively filled into the columns, from left to right in Figure \ref{fig:cap}, until the capacity $\kappa_j$ of resource $r_j$ is exhausted. The slots in the last column $l$ are occupied by the single $P$ command that does not occur in the particular program of that processor in previous columns.

\begin{figure}[h]
\begin{tikzpicture}
\draw (0,0) -- (0,5) -- (1,5) -- (1,0) -- (0,0);
\draw (0,1) -- (1,1);
\draw (1,5) -- (2,5) -- (2,0) -- (1,0);
\draw (1,2) -- (2,2); 
\draw (4,0) -- (4,5) -- (5,5) -- (5,0) -- (4,0);
\draw (4,4) -- (5,4);
\draw (5,5) -- (6,5) -- (6,0) -- (5,0);
\draw (5,4) -- (6,4);
\draw (5,3) -- (6,3);
\draw (5,1) -- (6,1);
\draw (5,2) -- (6,2);
\node at (0.5,4.5) {$P_1$};
\node at (0.5,2.5) {{\tiny $\bar{c}_1$}};
\node at (0.5,0.7) {$P_2$};
\node at (0.5,0.2) {{\tiny $\bar{d}_1$}};
\node at (1.5,4.5) {$P_2$};
\node at (1.5,2.5) {{\tiny $\bar{c}_2$}};
\node at (1.5,1.7) {$P_3$};
\node at (1.5,0.2) {{\tiny $\bar{d}_2$}};
\node at (3,2.5) {\dots};
\node at (4.5,4.7) {$P_{l-1}$};
\node at (5.5,4.7) {$P_l$};
\node at (4.5,4.2) {{\tiny $\bar{c}_{l-1}$}};
\node at (5.5,4.2) {{\tiny $\bar{c}_{l-1}$}};
\node at (4.5,3.7) {$P_l$};
\node at (4.5,0.2) {{\tiny $\bar{d}_{l-1}$}};
\node at (5.5,0.7) {$P_1$};
\node at (5.5,0.2) {{\tiny $n-\kappa_1$}};
\node at (5.5,3.7) {$P_{l-1}$};
\node at (5.5,3.2) {{\tiny $n-\kappa_{l-1}$}};
\node at (5.5,1.7) {$P_2$};
\node at (5.5,1.2) {{\tiny $n-\kappa_2$}};
\node at (0.5,5.2) {$1$};
\node at (1.5,5.2) {$2$};
\node at (3,5.2) {\dots};
\node at (4.5,5.2) {$l-1$};
\node at (5.5,5.2) {$l$};
\node at (5.5,2.5) {\dots};
\end{tikzpicture}\caption{Start of a PV program on $n$ processors (their programs occur horizontally; only $P$ commands are shown) and $l$ resources.}\label{fig:cap}
\end{figure}
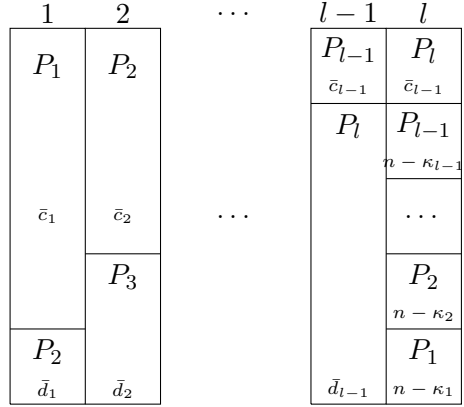

Let $v_0=(l-1,\dots ,l-1)$ denote the vertex corresponding to the final $P$ commands, ie after column $l-1$. Then $cr_i(v_0)=\kappa (r_i), i<l$: these capacities are exhausted. Moreover, $cr_l(v_0)=\bar{d}_{l-1}$, and hence $\kappa(X;v_0)=\kappa_l-\bar{d}_{l-1}=\bar{c}_l=\sum_1^l\kappa_l-(l-1)n$ by (\ref{eq:1}). The result follows from Proposition \ref{prop:mini} and Corollary \ref{cor:mini} since all vertices below $v_0$ have higher capacities and since future links of all vertices above $v_0$ until the final vertex $v_1=(2l+1,\dots ,2l+1)$ are contractible, cf Lemma \ref{lem:contr}.
\end{example}

\subsection{Spare capacity in the presence of crashes}
What happens to the spare capacity if one (or several) of the processors crashes during a computation? Potential crashes play an important role in the analysis of the resilience of a concurrent program in Distributed Computing \cite{HKR:14}.

For simplicity of notation, we assume that the \emph{last} processor $n$ crashes. First of all, state space is reduced: The interval $[0,l(k_n)+1]$ is reduced to $[0,C]$ with $0<C<l(k_n)+1$. If $\pi_n: X\to\mb{R}$ denotes projection to the last component, then the smaller state space is $X_C=\pi_n^{-1}([0,C])\subset X=\pi_n^{-1}([0,l(k_n)+1])=X$. This makes all vertices $v$ with $v_n>C$ irrelevant. On the other hand, new vertices on the \emph{upper boundary} $x_j=C$ appear. We let $B$ denote the largest integer strictly smaller than $C$  and such that $B\in Pr([1:k_n])$ for some $r\in R$ -- ie the command issued at $x_n=B$ is $Pr$. We compare the connectivities of future links at vertices $v_C=(v_1,\dots ,v_{n-1},C)\in (X_C)_0^P$ and $v_B=(v_1,\dots ,v_{n-1},B)\in X_0^P$.

\begin{lemma}
If $lk^+(X,v_B)$ is $k$-connected, then $lk^+(X_C,v_C)$ is at least $(k-1)$-connected.\\ If the future link at $v_C$ is not $k$-connected, then $lk^+(X,v_B)$ is not $(k+1)$-connected and the crash happens before resource $r$ was relinquished again.
\end{lemma}

\begin{proof}
In view of Proposition \ref{cor:linkconn}, we compare the spare capacities and thus the consumption functions at $v_B$ and $v_C$. For a resource $r'\neq r$, we have $cr'(v_C)\le cr'(v_B)$ (the resource might have been relinquished between $B$ and $C$). On the other hand, $cr(v_B)\le cr(v_C)\le cr(v_B)+1$ (depending on whether $r$ was relinquished between $B$ and $C$). As a consequence (cf Definition \ref{def:cc}), $\kappa (X;v_C)\ge\kappa (X;v_B)-1$. Equality can only hold if $cr(v_C)=cr(v_B)+1$, ie, if resource $r$ has not been relinquished between $B$ and $C$. 
\end{proof}

This result allows us to compare the spare capacity $\kappa(X)$ of a concurrent program without crash with the spare capacity $\kappa(X_C)$ of that program with a crash at $x_n=C$ (used as estimate for connectivities of path spaces, cf Section \ref{ss:sparecap}):

\begin{proposition}\label{prop:vertexcrash}
The spare capacity of $X_C$ is at least $\kappa(X_C)\ge\kappa(X)-1$.\\ If it is less than $\kappa(X)$, then there exists a vertex $v_B=(v_1,\dots , v_{n-1},B)\in X_0^P$ with $\kappa (X)=\kappa (X;v_B)$, ie of minimal spare capacity, a resource $r$ and an integer $l\le k_n(r)$ such that $Pr(l)=B$ \emph{(}the command issued at $x_n=B$ is of type $Pr$\emph{)} and $Vr(l)>C$ \emph{(}ie that resource would be relinquished after the crash\emph{)}.
\end{proposition}

\begin{example}
Consider the following simple example resulting in a drop of spare capacity and hence connectivity after a crash: Three processors compete for a single use of one resource $r$ of capacity $\kappa (r)= 2$ resulting in spare capacity $\kappa(X)=\kappa(X;v_B)=2-0=2$. If one of the processors crashes while having acquired a lock, the remaining two processors compete at $v_C$ with $\kappa(X_C)=\kappa(X_C;v_C)=2-1=1$, cf Figure \ref{fig:crash}.

\begin{figure}[h]
\begin{tikzpicture}[scale=0.7]
\draw[fill,color=pink] (-1.3,1.3) rectangle (4.4,7);
\draw (0,0) rectangle (5.7,5.7);
\draw (0,5.7) -- (-2,7.7) -- (3.7,7.7);
\draw (4.4,1.3) -- (5.7,0) -- (5.7,5.7) -- (3.7,7.7) -- (3.7,7); 
\draw[dashed] (3.7,7) -- (3.7,2);
\draw[dashed] (3.7,2) -- (4.4,1.3);
\draw[dashed] (-2,2) -- (3.7,2);
\draw (0,0) -- (-2,2) -- (-2,7.7);
\draw[fill,color=gray] (1.5,1.2) rectangle (3.5,4.8);
\draw[fill,color=gray] (1.5,4.8) -- (0.9,5.4) -- (2.9,5.4) -- (3.5,4.8) -- (1.5,4.8);
\draw[fill,color=gray] (3.5,1.2) -- (3.5,4.8) -- (2.9,5.4) -- (2.9,1.8) -- (2.9,1.2);
\draw [fill,color=gray] (1.5,1.2) -- (0.9,1.8) -- (0.9,5.4) -- (1.5,4.8) -- (1.5,1.2); 
\draw [fill, color=red] (1.5,1.2) circle (0.1cm);
\draw [fill, color=red] (0.9,1.8) circle (0.1cm);
\draw node at (1,1.2) {$v_B$};
\draw node at (0.3,1.8) {$v_C$};
\draw node at (-1.3,1.3) {$C$};
\end{tikzpicture}
\caption{Spare capacities -- with $X_C$ in front of the crash wall (in pink): $\kappa(X;v_B)=2, \kappa(X_C;v_C)=1$.}
\label{fig:crash}
\end{figure}
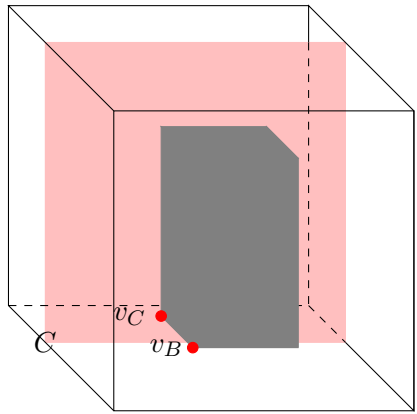
\end{example}

\begin{remark}
In distributed networks, individual processors are usually not aware that one of their partners has crashed. If crashes are known (for a resource), then locks acquired by the crashed processor can be deleted, possibly increasing spare capacities.\\
Remark that crashes also restrict the possible target vertices!
\end{remark}

\subsection{Programs with branches and loops}\label{ss:braloo}
Allowing branches and loops in concurrent programs changes the set-up, but not dramatically. Every single processor $p_i$ proceeds then along (the geometric realization of) a \emph{directed graph} $G_i$ instead of along an interval. A concurrent execution corresponds to a directed path in the product $\prod_1^n G_i$ of graphs from which certain forbidden regions have to be removed.

The space of directed paths in any of the graphs $G_i$ from source to target is homotopy discrete: Every connected component of the path space can be represented by a constant speed path $p_i$ with a directed interval within $G_i$ as range; other directed paths in this component are reparametrizations of the given one. The space of directed paths in $\prod_1^n G_i$ (without considering conflicting resource consumption) is a disjoint union of the spaces of directed paths corresponding to $n$-tuples of such components, represented by $n$-tuples of paths $(p_1,\dots ,p_n)$ -- an ``unfolding''. The space of directed paths corresponding to a particular unfolding can then be described via directed paths in a \emph{product of intervals}, and analysed  as in the previous sections. We are led to capacities and connectivity bounds that will often differ among the unfoldings, ie among the homotopy classes of directed paths in $\prod G_i$. 

\section{Algorithmics}
\subsection{Notation}\label{ss:notat}
The start data consist of a $PV$ program on $n$ threads and $l$ resources $r^i\in R$.
For  a non-empty subset $R'\subseteq R$, let 
\begin{itemize}
\item $P^j_{R'}(X):=\{l\in [1:l(j)]+1|\; \exists r\in R', i\in [1:k_j(r)]|\; l=Pr_j(i)\}\cup\{ l(j)+1\}\subset [1:l(j)+1]$ denote the subset of $P$-calls from $j\in P$ to a resource $r\in R'$ including the final position; with predecessor function\\ $p^j_{R'}: P^j_{R'}\to P^j_{R'}\cup\{ 0\},\; p^j_{R'}(k)=\max\{ l\in P^j_{R'}\cup\{ 0\}|\; l<k\}$;
\item $V^j_{R'}(X):=\{l\in [1:l(j)]|\; \exists r\in R', i\in [1:k_j(r)]|\; l=Vr_j(i)\}\subset [1:l(j)]$ denote the subset of $V$-calls from $j\in P$ to a resource $r\in R'$; with successor function\\ $s^j_{R'}:P^j_{R'}\to V^j_{R'},\; s^j_{R'}(k)=\min\{ l\in V^j_{R'}|\; k<l\}$.
\end{itemize}

We introduce the following integer (sub-)lattices in $\mb{R}^P=\mb{R}^n$:
\begin{itemize}
\item $L(X):=\prod_{j\in P}[0:l(j)+1]$;
\item $P_{R'}(X):=\prod_{j\in P}P^j_{R'}(X)\subset \bar{P}_{R'}(X):=\prod_{j\in P}(P^j_{R'}(X)\cup\{ 0\})\subset L(X)$; with predecessor function $p_{R'}: P_{R'}(X)\to\bar{P}_{R'}(X),\; p_{R'}([v_1,\dots ,v_n])=[p^1_{R'}(v_1),\dots ,p^n_{R'}(v_n)]$;
\item $V_{R'}(X):=\prod_{j\in P}V^j_{R'}(X)\subset L(X)$; with successor function $s_{R'}: P_{R'}(X)\to V_{R'}(X),$\\ $s_{R'}([v_1,\dots , v_n])=[s^1_{R'}(v_1),\dots ,s^n_{R'}(v_n)]$;
\end{itemize}

Capacities $\kappa (r^i)$ of individual resources $r^i\in R$ are collected in an $l=|R|$-dimensional \emph{capacity vector} $\bm{\kappa}=[\kappa r^1,\dots ,\kappa r^l]\in\mb{N}^l$.
To a grid vector $\mb{v}=[v_1,\dots ,v_n]\in P(X)=P_R(X)$, associate resource \emph{consumption vectors} 
\begin{itemize}
\item $\mb{c}(\mb{v})=[cr^1(\mb{v}),\dots ,cr^l(\mb{v})]\in\mb{N}^l_{\ge 0}$ (resource consumption ``at/just before'' $\mb{v}$)
\item $\mb{d}(\mb{v})=[dr^1(\mb{v}),\dots ,dr^l(\mb{v})]\in\mb{N}^l_{\ge 0}$ (``new'' locks asked for at $\mb{v}$; all $dr^i(\mb{v})\ge 0$; often $\sum_1^l dr^i(\mb{v})=n$).
\item $\mb{l}(\mb{v})=[l^1(\mb{v}),\dots ,l^l(\mb{v})]=\mb{c}(\mb{v})+\mb{d}(\mb{v})$ (resource consumption ``right after'' $\mb{v}$)
\end{itemize}
For a non-empty subset $R'\subset R$, the projection $\mb{Z}^R\to\mb{Z}^{R'}$ maps $\bm{\kappa}, \mb{c}(\mb{v}), \mb{d}(\mb{v})$ to $\bm{\kappa}_{R'}, \mb{c}_{R'}(\mb{v})$ and $\mb{d}_{R'}(\mb{v})$.
This is particularly relevant for the subset $R'=R(\mb{v}):=\{ r\in R|\; dr(\mb{v})>0\}\subseteq R$, the set of resources actually called for at a vertex $\mb{v}$. 

\subsection{Determining spare capacities algorithmically}\label{ss:sparecapalg}
In this section, we will only need an implementation of the grid/array $P(X):=P_R(X)\subset L(X)$. We assume throughout that vectors, including vector addition and dot product, are implemented on that entire array.
\subsubsection{Calculations required for a single processor}
To initialize, consider just a single processor $j\in P$ : We defined in Section \ref{ss:consump} the difference function $dr_j: [0:l(j)+1]\to\{ 0, 1, -1\}$ for every resource $r\in R$. Combined over all resources $r\in R$, they define a difference vector function $\mb{d}_j: [0:l(j)+1]\to\{ 0, 1, -1\}^R$. It can be read from the command line for processor $j$ in $l(j)$ steps resulting in $|R|$-dimensional vectors $\mb{d}_j(k),\; 0\le k\le l(j)+1$. If only one call is issued at every node, then $\mb{d}_j$ takes signed \emph{unit} vectors as values; at the ends $0$ and $l(j)+1$, it takes the fixed value $\mb{0}$.  

Resource consumption $cr_j: [0:l(j)]\to \{ 0,1\}$ (interpreted as Boolean values) is calculated inductively: There are only two cases in which $cr_j(k-1)\neq cr_j(k)$: That happens if $dr_j(k-1)=1$ and $dr_j(k)=0$ (flipping $cr_j$ from $0$ to $1$) or if $dr_j(k-1)=0$ and $dr_j(k)=-1$ (flipping $cr_j$ from $1$ to $0$); equivalently, if $dr_j(k-1)-dr_j(k)=1$. Taking these difference calculations and conditional flips over all resources $r\in R$  and determining the entire resource consumption function $\mb{c}_j: [0:l(j)]\to \{ 0,1\}^R$ takes thus $2l(j)$ steps.

As a result, establishing the resource consumption functions $\mb{c}_j$ for every processor $j\in P$ requires in total 3$\sum_1^nl(j)$ steps.

\subsubsection{Array calculations: spare capacities of vertices}\label{sss:sparecapvert}

In the next steps, calculate the vector functions $\mb{d}: P(X)\to\{ 0,1\}^R$ by $\mb{d}(\mb{v})=\sum_{j\in P}\mb{d}_j(v_j)$ and $\mb{c}: P(X)\to (\mb{N}_{\ge 0})^R$ by $\mb{c}(\mb{v})=\sum_{j\in P}\mb{c}_j(v_j)$. Both are defined on the restricted grid $P(X)$ only and hence $\mb{d}$ -- with information on ``new calls'' $Pr$ -- takes only values $0$ or $1$ at each component. Each of these calculations requires $n+1=|P|+1$ steps: After initializing the entire array with $0$-vectors, the \emph{same} vector functions $\mb{d}_j$, resp.\ $\mb{c}_j,\; j\in P$, are added to all cells (in different ``directions'' $j\in P$, of course). 

Next, calculate the difference $\bm{\kappa}-\mb{c}: P(X)\to\mb{Z}^R$ of the fixed capacity vector $\bm{\kappa}\in\mb{N}^R$ and the capacity function $\mb{c}$. Every vertex $\mb{v}$ for which $\bm{\kappa}-\mb{c}(\mb{v})$ has at least one negative component ($|R|$ comparison steps) belongs to the forbidden region and is flagged.

Finally, spare capacities for (non-flagged) vertices $\mb{v}$ are determined by one dot product operation $\bm{\kappa}(X;\mb{v})=\mb{d}(\mb{v})\cdot (\bm{\kappa}-\mb{c}(\mb{v}))\in\mb{Z}$ on the entire array $P(X)$. Under our assumptions, the total number of steps is thus \emph{linear} in the number of $PV$-steps on each of the processors in $P$ and on the number of resources in $R$. 

\subsubsection{Determining the spare capacity of a complex associated to a $PV$ program}
In a final round, probe successively equations $\bm{\kappa}(X;\mb{v})=k$ starting with and increasing from $k=0$ on the (non-flagged) vertices of the entire array. The minimal $k$ for which an equation $\bm{\kappa}(X;\mb{v})=k$ yields the answer \texttt{true} for some vertex $\mb{v}\in P(X)$ corresponds to the spare capacity $\kappa(X)$ of the state space. The number of steps needed is at most $n$ for a $PV$ program on $n$ processors: Future links are all contained in an $(n-1)$-simplex. If spare capacities are at least $n-1$ at \emph{every} vertex, then these future links are full simplices, and all path spaces are contractible.


\subsection{Deadlocks and doomed regions}\label{ss:deaddoom}
If a PV program leads to one or several deadlocks (ie to a vertex $v$ with $\kappa (X;v)=0$, cf. Proposition \ref{cor:linkconn}(3), then the spare capacity of the entire program $\kappa (X)$ vanishes as well, by definition. But it may be important to determine the spare capacity of the state space that arises outside the ``doomed regions'' (no directed path from there to the top vertex) associated with them. 

In the remaining two sections, particular consideration will be given to vertices $v$ with spare capacity $\kappa(X;v)=0$ (ie deadlocks) and those with $\kappa(X;v)=1$ (ie with disconnected future link $lk^+(X,v)$). 

\subsubsection{Deadlock detection}\label{sss:dead} Deadlock detection for PV programs was provided in detail in Fajstrup etal.\ \cite{FGR:98} only in the case where all par\-ti\-ci\-pating resources $r\in R$ have capacity $\kappa (r)=n-1=|P|-1$. In a way, the general case, with other and variable capacities, could still be handled, since the forbidden region $F(r)$ associated to a resource $r$ of smaller capacity can be modelled as the union of many resources of capacity $n-1$. 
But it is preferable to give a formulation for deadlocks in the general case, with resources of various capacities participating: In view of Proposition \ref{cor:linkconn}(3), a vertex $v\in P(X)$ is a deadlock, if 
\begin{itemize}
\item $\mb{c}(v)-\bm{\kappa}\ge\mb{0}$ (ie $v$ is not contained in the forbidden region)
\item $\mb{c}_{R(v)}(v)=\bm{\kappa}_{R(v)}$ (All resources asked for at $v$ have been locked already before up to full capacity; none of the processors can advance).
\end{itemize}
A deadlock at $v$ -- detected via its capacity $\kappa (X;v)=0$ as in Section \ref{sss:sparecapvert} -- comes thus with the following data:
\begin{itemize}
\item A subset $R'\subseteq R$ of resources (ie $R'=R(v))$;
\item For each $r\in R'$, a subset $C(r)\subseteq P$ of locking processors of \emph{cardinality} $|C(r)|=\kappa (r)$ and a non-empty subset $D(r)\subseteq P\setminus C(r)$ of its complement such that $\bigsqcup_{r\in R'} D(r)=P=[1:n]$ is a partition of $P$ (ie the $D(r)$ are disjoint);
\item Every processor $j\in C(r),\; r\in R',$ has delivered a call $Pr$ in front of and still active at $v_j$ (ie $r\in R', j\in C(r) \Rightarrow cr_j(v_j)=1,$ hence $cr(v)\ge\kappa (r)$);
\item For every $j\in D(r),\; r\in R'$, a call $Pr$ is issued at $v_j$ (ie $dr_j(v)=1, cr_j(v)=0$, and hence $cr(v)=\kappa (r),\; r\in R'$);
\item For every $r'\in R\setminus R'$, at most $\kappa (r')$ calls $Pr'$ are active at $v$ (ie $cr'(v)\le\kappa (r'),\; r'\in R\setminus R'$).
\end{itemize}
The last requirement makes sure that $v$ is not a forbidden vertex.
\subsubsection{Doomed regions}\label{sss:doomed}
Compare with Fajstrup etal \cite{FGR:98} (where these regions are called ``unsafe'') and Fajstrup etal \cite{FGHMR:16}.

Let $v$ denote a deadlock vertex with resource calls to $R(v)\subseteq R$ and predecessor vertex $w:=p_{R(v)}(v)$, cf Section \ref{ss:notat}. Remark that 
$\mb{l}_{R(v)}(w)=\mb{c}_{R(v)}(w)+\mb{d}_{R(v)}(w)=\mb{c}_{R(v)}(v)=\bm{\kappa}_{R(v)}$. Hence, every resource $r\in R(v)$ is locked by $\kappa (r)$ processors within the  hyperrectangle $D(v):=]w,v]=\prod_{j=1}^n]v_j,w_j]$ spanned by $w$ and $v$, and no directed path can leave $D(v)$.  

One may eliminate this ``primary'' \emph{doomed region} $D(v)$ from the state space $X$ by a modification of the original $PV$-program: Add an extra resource $\bar{r}$ of capacity $n-1$ and, for each $j\in P$, calls $P\bar{r}$ at predecessors $w_j=p_{R(v)}^j(v_j)$ to be relinquished by calls $V\bar{r}$ at the successors $s_{R(v)}^j(v_j)$. Then $D(v)=F(\bar{r}):=]w,x],\; x=s_{R(v)}(v)$, becomes part of the forbidden region of the modified program -- but path spaces with target not included in $D(v)$ remain unchanged!

Adding $F(\bar{r})=]w,x]$ to the forbidden region $F$, one can, in the same way as described in \cite{FGR:98}, inductively define \emph{higher order} doomed regions: With the updated capacity consumption, new deadlocks may arise at the intersection of the boundaries of the doomed region $F(\bar{r})$ and the original forbidden region $F$. Modifying the recursive algorithm from \cite{FGR:98}, one obtains a program that is \emph{deadlockfree and with literally the same path spaces as before} - if just the target is not contained in any of the doomed regions (from which it cannot terminate correctly).

\subsection{Disconnected futures}
\subsubsection{Vertices with disconnected future links}\label{sss:disconn}
As a consequence of Corollary \ref{cor:linkconn}(4), a vertex $v$ has a disconnected non-empty future link $lk^+(X;v)$ if and only if
\begin{enumerate}
\item $\mb{c}(v)-\bm{\kappa}\ge\mb{0}$ (ie $v$ is not contained in the forbidden region);
\item $\mb{c}_{R(v)}(v)-\bm{\kappa}_{R(v)}$ is a standard unit vector $\mb{e}_{r^0}$ with $r^o\in R(v)$ (with a single coordinate $1$, all others $0$);
\item $dr^0(v)>1$ (at least two calls $Pr^0$ at $v$).
\end{enumerate}
Such a \emph{critical vertex} $v$ can be characterized by the following data (this is just a small variation compared to the characterization of deadlocks in Section \ref{sss:dead}): 
\begin{itemize}
\item A subset $R'\subseteq R$ of resources including a particular element $r^0\in R'$ (ie $R'=R(v)$);
\item For each $r\in R'$ a subset $C(r)\subseteq P$ of processors such that $|C(r^0)|=\kappa (r^0)-1, |C(r)|=\kappa (r), r\in R'\setminus\{r^0\}$, and a non-empty subset $D(r)\subseteq P\setminus C(r)$ of its complement such that $|D(r^0)|\ge 2$ and $\bigsqcup_{r\in R'} D(r)=P=[1:n]$ is a partition of $P$ (ie the $D(r)$ are disjoint);
\item Every processor $j\in C(r),\; r\in R',$ has delivered a call $Pr$ in front of and still active at $v_j$ (ie $ r\in R' \Rightarrow cr_j(v_j)=1$, hence $cr(v)\ge\kappa (r), r\neq r^0, cr^0(v)\ge\kappa (r^0)-1$);
\item For $j\in D(r)$, a call $Pr$ is issued at $v_j$ (ie $dr_j(v)=1, cr_j(v)=0$, hence $cr^0(v)=\kappa (r^0)-1,$ $cr(v)=\kappa (r),\; r\neq r^0$);
\item For every $r'\in R\setminus R'$, at most $\kappa (r')$ calls $Pr'$ are active at $v$ (ie $cr'(v)\le\kappa (r'),\; r'\in R\setminus R'$).
\end{itemize}

\subsubsection{Doomed region for disconnectivity}\label{sss:doomdis}
Analogous to the doomed region $D(v)$ below a deadlock vertex $v$ from Section \ref{sss:doomed}, there is a \emph{critical} region $D^1(v)$ below a vertex $v$ with disconnected future link such that path spaces $\vec{P}(X)_y^t$ are \emph{disconnected} for $y\in D^1(v)$:

 Let $v$ denote a vertex satisfying the conditions in Section \ref{sss:disconn}, with resource calls to $R(v)\subseteq R$ and predecessor vertex $w=p_{R(v)}(v)$. Define $D^1(v):=]w,v]$. Note that a directed path can leave this hyperrectangle only through a hyperplane $x_i=v_i$ for $i\in C(r^0)$; across all other upper boundary hyperplanes $x_j=v_j, j\not\in C(r^0)$, it would enter the forbidden region. Moreover, such a directed path can enter $x_i>v_i$ for \emph{only one} $i\in C(r^0)$ -- but not both $x_i>v_i$ and $x_j>v_j,\; i\neq j,$ -- without entering the forbidden region.

 Let $x=s_{R(v)}(v)\in V_{R(v)}(X)\subset L(X)$ denote the successor vertex of $v$ with respect to $R'=R(v)$; cf Section \ref{ss:notat}. The intersection of the state space $X$ with $]w,x[\setminus ]w,v]$ -- that every directed path from $D^1(v)$ needs to enter -- has the form $\coprod_{i\in C(r^0)}(]v_i,x_i]\times\prod_{j\neq i}]w_j,v_j[).$ Remark that the subspaces in that disjoint union are not connected to each other: \emph{exactly one} coordinate is larger than $v_i$. 

As in the case of a doomed region for a deadlock, one may eliminate $]w,x[$ -- and hence the critical region $D^1(v)$ -- from the state space by adding an additional resource of capacity $n-1$ which is locked, for each processor $j\in P$, at $w_j$ and relinquished at $x_j$. The arising new state space may contain further deadlocks at the intersection of old and new forbidden regions: the associated doomed regions consist of those points $u$ such that every directed path starting at $u$ needs to pass through the critical region $D^1(v)$ -- with disconnected path spaces (with source $u$) as a consequence; these new doomed regions are \emph{higher order critical regions} with respect to the vertex $v$. Eliminating \emph{all} critical regions and associated higher order critical regions results in a state space $\tilde{X}\subseteq X$ with all spaces of directed paths between vertices being path-connected.

\subsubsection{Estimation of the number of path components}
Mutually reachable critical vertices (with spare capacity $1$), or rather their future links allow determining an upper bound to the number of path components of the space of directed paths between vertices: For every critical vertex $c\in X_0$, consider the connected components $c_i$ of its future link (there are at most $n$ of them) and the partial order relation $\preceq$ given by reachability (within $X$) between components of future links (each of them representing an edge) of various critical vertices. A (possibly empty) chain of components of future links (between  a given source $s$ and target vertex $t$) can be realized by a directed path since reachability was assumed. It is known that every directed path in $\vec{P}(X)_s^t$ is d-homotopic to a \emph{tame} directed path with source $s$ and target $t$ (Ziemia\'{n}ski \cite[Prop.\ 6.28]{Ziemianski:12},\cite[Theorem 5.6]{Ziemianski:17}, Raussen \cite[Theorem 2.6]{Raussen:21}) that can only transit from one cube to another at a vertex. If a tame path enters a critical region $D^1(v)$, it has to leave it at its top vertex $c$ and then along one of the 1-cubes (edges) $c_i$.

Moreover, two directed paths realizing a chain of components (no other critical vertices and critical regions involved!) are d-homotopic to each other. This can be seen by a minor modification of the proof of Corollary \ref{cor:conn}: At every critical vertex, only one of the possible future components is allowed. Excluding deadlocks and non-selected components of future links, all remaining vertices have a spare capacity at least $2$. 

Hence, the number of path components of $\vec{P}(X)_s^t$ can be estimated (from above) by the number of chains described above. If this number is not too large, the possible outcomes of \emph{all} executions can thus be determined by running one execution along every such chain. 


\end{document}